\documentclass[letterpaper, 10 pt, conference]{ieeeconf}%

\IEEEoverridecommandlockouts
\usepackage{cite}
\usepackage{amsmath,amssymb,amsfonts}
\usepackage{algorithmic}
\usepackage{graphicx}
\usepackage{textcomp}
\usepackage{xcolor}
\usepackage{cuted}
\usepackage{lipsum}
\usepackage{mathtools}
\usepackage{stfloats}
\usepackage{bbm}
\usepackage{multirow}

\newcommand{\R}{\mathbb{R}}
\newcommand{\N}{\mathbb{N}}
\newcommand{\T}{\top}
\newcommand{\I}{\mathbf{I}}
\newcommand{\0}{\mathbf{0}}
\newcommand{\E}{\mathcal{E}}

\newcommand{\diag}{\text{diag}}
\newcommand{\tsup}[1]{\textsuperscript{#1}}
\newcommand{\mb}[1]{\mathbf{#1}}

\newcommand{\bm}[1]{\begin{bmatrix}#1\end{bmatrix}}
\def\BibTeX{{\rm B\kern-.05em{\sc i\kern-.025em b}\kern-.08em
    T\kern-.1667em\lower.7ex\hbox{E}\kern-.125emX}}

\newtheorem{assumption}{\textbf{Assumption}}
\newtheorem{theorem}{\textbf{Theorem}}
\newtheorem{corollary}{\textbf{Corollary}}
\newtheorem{lemma}{\textbf{Lemma}}
\newtheorem{proposition}{\textbf{Proposition}}
\newtheorem{remark}{\textbf{Remark}}
\newtheorem{definition}{\textbf{Definition}}

\usepackage{hyperref}
    
\begin{document}

\title{
{\LARGE \textbf{
{Dissipativity-Based Distributed Control and Communication Topology Co-Design for DC Microgrids with ZIP Loads}
}}}



\author{Mohammad Javad Najafirad and Shirantha Welikala 
\thanks{The authors are with the Department of Electrical and Computer Engineering, School of Engineering and Science, Stevens Institute of Technology, Hoboken, NJ 07030, \texttt{{\small \{mnajafir,swelikal\}@stevens.edu}}.}}

\maketitle


\begin{abstract}
This paper presents a novel dissipativity-based distributed droop-free control and communication topology co-design approach for voltage regulation and current sharing in DC microgrids (DC MGs) with generic ``ZIP'' (constant impedance (Z), current (I) and power (P)) loads. While ZIP loads accurately capture the varied nature of the consumer loads, its constant power load (CPL) component is particularly challenging (and destabilizing) due to its non-linear form. Moreover, ensuring simultaneous voltage regulation and current sharing and co-designing controllers and topology are also challenging when designing control solutions for DC MGs.  
To address these three challenges, we model the DC MG as a networked system comprised of distributed generators (DGs), ZIP loads, and lines interconnected according to a static interconnection matrix. Next, we equip each DG with a local controller and a distributed global controller (over an arbitrary topology) to derive the error dynamic model of the DC MG as a networked ``error'' system, including disturbance inputs and performance outputs. Subsequently, to co-design the controllers and the topology ensuring robust (dissipative) voltage regulation and current sharing performance, we use the dissipativity and sector boundedness properties of the involved subsystems and formulate Linear Matrix Inequality (LMI) problems to be solved locally and globally. To support the feasibility of the global LMI problem, we identify and embed several crucial necessary conditions in the corresponding local LMI problems, thus providing a one-shot approach (as opposed to iterative schemes) to solve the LMI problems. Overall, the proposed approach in this paper provides a unified framework for designing DC MGs. The effectiveness of the proposed solution was verified by simulating an islanded DC MG under different scenarios, demonstrating superior performance compared to traditional control approaches.
\end{abstract}


\section{Introduction} 
The microgrid (MG) concept has been introduced as a comprehensive framework for the cohesive coordination of distributed generators (DGs), variable loads, and energy storage units within a controllable electrical network to facilitate the efficient integration of renewable energy resources, such as wind turbines and photovoltaic systems \cite{liu2023resilient}. In particular, DC MGs have gained more attention in recent years due to the growing demand for DC loads, such as data centers, electric vehicles, LED lighting systems, and consumer electronic devices. In addition, DC MGs offer distinct advantages over AC systems by eliminating unnecessary conversion stages and removing frequency regulation \cite{dou2022distributed}. However, the faster dynamics of DC MGs, compared to AC systems, demand the meticulous design of rapid and robust control systems. 

The two primary control goals in DC MGs are voltage regulation and current sharing. To achieve these goals, centralized \cite{mehdi2020robust}, decentralized \cite{peyghami2019decentralized}, and distributed \cite{xing2019distributed} control solutions have been proposed in the literature. Although the centralized control approach provides controllability and observability, it suffers from a single point of failure \cite{guerrero2010hierarchical}. In decentralized control, only a local controller is required; hence, there is no communication among DGs \cite{khorsandi2014decentralized}, which hinders coordination affecting current sharing. However, this lack of coordination can be solved by developing a distributed control solution where the DGs can share their state variables with their communicating neighbors through a communication network \cite{dehkordi2016distributed}.

It should be noted that the traditional and widely adopted droop control approach is a conventional decentralized control solution. However, due to line impedance mismatch and droop characteristics, droop control cannot simultaneously achieve voltage regulation and current sharing. Despite various innovations in hierarchical and distributed implementation strategies \cite{zhou2020distributed,nasirian2014distributed,najafirad1}, droop control fundamentally requires careful tuning of droop coefficients to balance these conflicting objectives effectively. Consequently, recent research has focused on developing ``droop-free'' distributed control algorithms that rely on communications among DGs \cite{dissanayake2019droop,zhang2022droop}, offering more flexibility in achieving voltage regulation and current sharing simultaneously.

A particular challenge for developing such distributed control solutions for DC MGs is the presence of generic ``ZIP'' (constant impedance (Z), constant current (I) and constant power (P)) loads. In particular, constant power loads (CPLs) exhibit negative impedance characteristics that can destabilize the system \cite{kwasinski2010dynamic}. Moreover, the nonlinear nature of CPLs introduces significant challenges to stability analysis and robust controller design, necessitating advanced control techniques to ensure accurate and robust operation of the DC MG in the presence of ZIP loads \cite{hassan}.

Furthermore, conventional distributed controller design proceeds independently from communication topology considerations, with network structures often assumed to be fixed and predefined. Recent advancements in communication technologies have eliminated the necessity for fixed network structures, creating opportunities to implement innovative control strategies exploiting customizable/reconfigurable communication topologies \cite{jin2017toward}. Several approaches to identify a cost-effective communication network that also ensures control performance have been proposed in \cite{hu2021cost,lou2018optimal,sheng2022optimal}. Still, such approaches typically use a sequential design process rather than a co-design strategy that equally respects control and communication goals.

To address the aforementioned challenges in designing distributed controllers for DC MGs, we use the dissipativity theory, which offers a powerful framework for analyzing and designing robust control systems for large-scale networked systems. The dissipativity concepts have already been used for many power system applications like power electronic converters and MGs \cite{datadriven}. By focusing on the fundamental energy exchanges between interconnected subsystems, dissipativity-based approaches can ensure the stability and robustness of the networked system even when subsystems exhibit complex, nonlinear behaviors \cite{arcak2022}. 

The relationship between communication network density and control performance is non-monotonic. Figure \ref{fig.initialresult} shows that excessive links provide limited performance improvement while increasing system complexity. This trade-off motivates co-design approaches that simultaneously optimize controller parameters and communication topology.

\begin{figure}
    \centering
    \includegraphics[width=0.99\columnwidth]{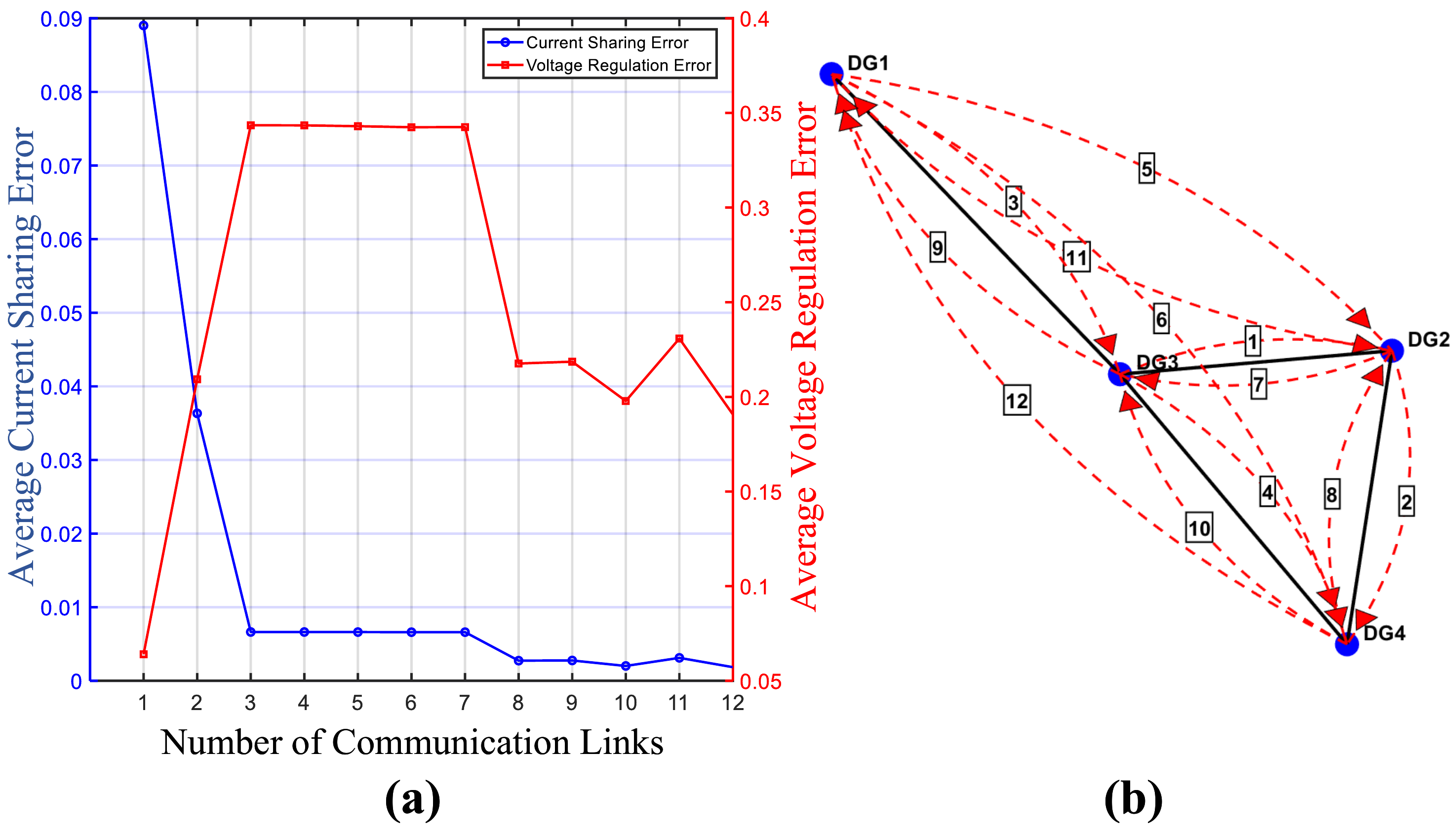}
    \vspace{-7mm}
    \caption{Communication network impact on DC MG: (a) performance metrics versus link count and (b) network structure.}
    \label{fig.initialresult}
\end{figure}

In this paper, we propose a dissipativity-based hierarchical distributed control framework to simultaneously achieve voltage regulation and current sharing goals in DC MGs with ZIP loads. First, we model the DC MG as a networked system of subsystems interconnected according to an interconnection matrix. Then, a hierarchical controller containing steady-state, local, and distributed global components is proposed for the DGs. Through conducting a steady state analysis, we identify the equilibrium state and show that the error dynamics of the DC MG (around the equilibrium point) can be represented as a networked system, which can also include disturbance inputs and performance outputs to ensure robust achievement of the desired control objectives. To co-design the controllers and topology, we exploit the dissipative nature of the involved subsystems and the sector-bounded nature of the CPLs to formulate linear matrix inequality (LMI) problems to be solved locally and globally. To support the global LMI problem's feasibility, we systematically identify and embed several necessary conditions into the corresponding local LMI problems. This integrated approach makes the overall design process a streamlined, one-shot procedure, avoiding the need for iterative schemes that demand complex computations and prolonged convergence times \cite{almihat2023overview}.

It is worth noting that, in our prior work on DC MGs \cite{ACCNajafi} (available at \cite{Najafirad2024Ax1}), we have focused exclusively on achieving voltage regulation, omitting the presence of CPLs and bypassing thorough steady state analysis and comprehensive necessary conditions to be enforced at DGs. Building upon the foundation established by \cite{ACCNajafi}, this paper presents a significantly more complete framework for designing control systems for DC MGs. The main contributions of this paper can be summarized as follows:
\begin{enumerate}
\item 
We formulate the DC MG control problem as a networked system control problem and propose a novel hierarchical control framework that combines local steady-state and voltage regulation controllers with distributed global consensus-based current sharing controllers to formulate a unified control strategy. 
\item 
We use a generic ZIP load model that includes CPLs, and provide a technique to handle the nonlinearities (and destabilizing negative impedance characteristics) introduced by CPLs, without sacrificing the LMI (convex optimization) form of the overall control design framework.
\item 
We formulate the overall control and topology co-design problem as a set of local and global LMI (convex optimization) problems to be executed in one shot, enabling efficient and scalable numerical implementations. We have also identified a stronger set of necessary conditions that can be implemented in the local LMI problems so as to support the feasibility of the global LMI problem (i.e., of the co-design).
\end{enumerate}

The remainder of this paper is structured as follows. The essential concepts of dissipativity theory and networked systems establish the theoretical foundation in Sec. \ref{Preliminaries}. Our DC MG model with detailed physical topology and component dynamics appears in Sec. \ref{problemformulation}. Within Sec. \ref{Sec:Controller}, we introduce a novel hierarchical control architecture that eliminates traditional droop mechanisms. A comprehensive error dynamics framework addressing CPL nonlinearities follows in Sec. \ref{Sec:ControlDesign}. The development of our dissipativity-based methodology for controller and communication topology synthesis occupies Sec. \ref{Passivity-based Control}. Numerical simulations demonstrating the efficacy of our approach are presented in Sec. \ref{Simulation}. Finally, Sec. \ref{Conclusion} offers concluding remarks on contributions and directions for future research.

\section{Preliminaries}\label{Preliminaries}

\subsection{Notations}
The notation $\mathbb{R}$ and $\mathbb{N}$ signify the sets of real and natural numbers, respectively. 
For any $N\in\mathbb{N}$, we define $\mathbb{N}_N\triangleq\{1,2,..,N\}$.
An $n \times m$ block matrix $A$ is denoted as $A = [A_{ij}]_{i \in \mathbb{N}_n, j \in \mathbb{N}_m}$. Either subscripts or superscripts are used for indexing purposes, e.g., $A_{ij} \equiv A^{ij}$.
$[A_{ij}]_{j\in\mathbb{N}_m}$ and $\diag([A_{ii}]_{i\in\mathbb{N}_n})$ represent a block row matrix and a block diagonal matrix, respectively.
$\0$ and $\I$, respectively, are the zero and identity matrices (dimensions will be clear from the context). A symmetric positive definite (semi-definite) matrix $A\in\mathbb{R}^{n\times n}$ is denoted by $A>0\ (A\geq0)$. The symbol $\star$ represents conjugate blocks inside block symmetric matrices. $\mathcal{H}(A)\triangleq A + A^\T$,  $\mb{1}_{\{ \cdot \}}$ is the indicator function and $\mathbf{1}_N$ is a vector in $\R^N$ containing only ones.

\subsection{Dissipativity}
Consider a nonlinear dynamic system:
\begin{equation}\label{dynamic}
\begin{aligned}
    \dot{x}(t)=f(x(t),u(t)),\\
    y(t)=h(x(t),u(t)),
    \end{aligned}
\end{equation}
where $x(t)\in\mathbb{R}^n$, $u(t)\in\mathbb{R}^q$, $y(t)\in\mathbb{R}^m$, and $f:\mathbb{R}^n\times\mathbb{R}^q\rightarrow\mathbb{R}^n$ and $h:\mathbb{R}^n\times\mathbb{R}^q\rightarrow\mathbb{R}^m$ are continuously differentiable and $f(\0,\0)=\0$ and $h(\0,\0)=\0$.


\begin{definition}
The system (\ref{dynamic}) is $X$-dissipative if it is dissipative under the quadratic supply rate:
\begin{center}
$
s(u,y)\triangleq
\begin{bmatrix}
    u \\ y
\end{bmatrix}^\top
\begin{bmatrix}
    X^{11} & X^{12}\\ X^{21} & X^{22}
\end{bmatrix}
\begin{bmatrix}
    u \\ y
\end{bmatrix}.
$
\end{center}
\end{definition}

\begin{remark}\label{Rm:X-DissipativityVersions}
If the system (\ref{dynamic}) is $X$-dissipative with:
1)\ $X = \scriptsize\begin{bmatrix}
    \0 & \frac{1}{2}\I \\ \frac{1}{2}\I & \0
\end{bmatrix}\normalsize$, then it is passive;
2)\ $X = \scriptsize\begin{bmatrix}
    -\nu\I & \frac{1}{2}\I \\ \frac{1}{2}\I & -\rho\I
\end{bmatrix}\normalsize$, then it is strictly passive with input and output passivity indices $\nu$ and $\rho$, denoted as IF-OFP($\nu,\rho$);
3)\ $X = \scriptsize\begin{bmatrix}
    \gamma^2\I & \0 \\ \0 & -\I
\end{bmatrix}\normalsize$, then it is $L_2$-stable with gain $\gamma$, denoted as $L2G(\gamma)$; 
in an equilibrium-independent manner (see also \cite{WelikalaP42022}). 
\end{remark}

If the system (\ref{dynamic}) is linear time-invariant (LTI), a necessary and sufficient condition for $X$-dissipative is provided in the following proposition as a linear matrix inequality (LMI) problem.

\begin{proposition}\label{Prop:linear_X-EID} \cite{welikala2023platoon}
The LTI system
\begin{equation*}\label{Eq:Prop:linear_X-EID_1}
\begin{aligned}
    \dot{x}(t)=Ax(t)+Bu(t),\quad
    y(t)=Cx(t)+Du(t),
\end{aligned}
\end{equation*}
is $X$-dissipative if and only if there exists $P>0$ such that
\begin{equation*}\label{Eq:Prop:linear_X-EID_2}
\scriptsize
\begin{bmatrix}
-\mathcal{H}(PA)+C^\top X^{22}C & -PB+C^\top X^{21}+C^\top X^{22}D\\
\star & X^{11}+\mathcal{H}(X^{12}D)+D^\top X^{22}D
\end{bmatrix}
\normalsize
\geq0.
\end{equation*}
\end{proposition}



\subsection{Networked Systems} \label{SubSec:NetworkedSystemsPreliminaries}

Consider the networked system $\Sigma$ in Fig. \ref{Networked}, consisting of dynamic subsystems $\Sigma_i,i\in\mathbb{N}_N$, $\Bar{\Sigma}_i,i\in\mathbb{N}_{\Bar{N}}$ and a static interconnection matrix $M$ that characterizes interconnections among subsystems, exogenous inputs $w(t)\in\mathbb{R}^r$ (e.g. disturbances) and interested outputs $z(t)\in\mathbb{R}^l$ (e.g. performance). 

\begin{figure}
    \centering
    \includegraphics[width=0.6\columnwidth]{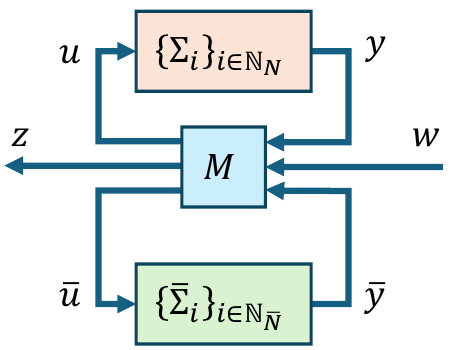}
    \caption{A generic networked system $\Sigma$.}
    \label{Networked}
\end{figure}

The dynamics of each subsystem $\Sigma_i,i\in\mathbb{N}_N$ are given by
\begin{equation}
    \begin{aligned}
        \dot{x}_i(t)=f_i(x_i(t),u_i(t)), \quad y_i(t)=h_i(x_i(t),u_i(t)),
    \end{aligned}
\end{equation}
where $x_i(t)\in\mathbb{R}^{n_i}$, $u_i(t)\in\mathbb{R}^{q_i}$, $y_i(t)\in\mathbb{R}^{m_i}$ and $f_i(\0,\0)=\0$ and $h_i(\0,\0)=\0$. In addition, each subsystem $\Sigma_i$ is assumed to be $X_i$-dissipative , where $X_i \triangleq [X_i^{kl}]_{k,l\in\N_2}$.
Regarding each subsystem $\bar{\Sigma}_i, i\in\N_{\bar{N}}$, we use similar assumptions and notations, but include a bar symbol to distinguish between the two types of subsystems, e.g.,  $\bar{\Sigma}_i$ is assumed to be $\bar{X}_i$-dissipative where $\bar{X}_i \triangleq [\bar{X}_i^{kl}]_{k,l\in\N_2}$.

Defining $u\triangleq[u_i^\top]^\top_{i\in\mathbb{N}_N}$, $y\triangleq[y_i^\top]^\top_{i\in\mathbb{N}_N}$, $\Bar{u}\triangleq[\Bar{u}_i^\top]^\top_{i\in\mathbb{N}_{\Bar{N}}}$ and $\Bar{y}\triangleq[y_i^\top]^\top_{i\in\mathbb{N}_{\Bar{N}}}$, the interconnection matrix $M$ and the corresponding interconnection relationship are given by
\begin{equation}\label{interconnectionMatrix}
\scriptsize
\begin{bmatrix}
    u \\ \bar{u} \\ z
\end{bmatrix}=M
\normalsize
\scriptsize
\begin{bmatrix}
    y \\ \bar{y} \\ w
\end{bmatrix}
\normalsize
\equiv
\scriptsize
\begin{bmatrix}
    M_{uy} & M_{u\bar{y}} & M_{uw}\\
    M_{\bar{u}y} & M_{\bar{u}\bar{y}} & M_{\bar{u}w}\\
    M_{zy} & M_{z\bar{y}} & M_{zw}
\end{bmatrix}
\begin{bmatrix}
    y \\ \bar{y} \\ w
\end{bmatrix}.
\normalsize
\end{equation}

The following proposition exploits the $X_i$-dissipative and $\bar{X}_i$-dissipative properties of the subsystems $\Sigma_i,i\in\mathbb{N}_N$ and $\Bar{\Sigma}_i,i\in\mathbb{N}_{\Bar{N}}$ to formulate an LMI problem for synthesizing the interconnection matrix $M$ (\ref{interconnectionMatrix}), ensuring the networked system $\Sigma$ is $\textbf{Y}$-dissipative for a prespecified $\textbf{Y}$ under two mild assumptions \cite{welikala2023non}.

\begin{assumption}\label{As:NegativeDissipativity}
    For the networked system $\Sigma$, the provided \textbf{Y}-dissipative specification is such that $\textbf{Y}^{22}<0$.
\end{assumption}

\begin{remark}
Based on Rm. \ref{Rm:X-DissipativityVersions}, As. \ref{As:NegativeDissipativity} holds if the networked system $\Sigma$ must be either: (i) L2G($\gamma$) or (ii) IF-OFP($\nu,\rho$) with some $\rho>0$, i.e., $L_2$-stable or passive, respectively. Therefore, As. \ref{As:NegativeDissipativity} is mild since it is usually preferable to make the networked system $\Sigma$ either $L_2$-stable or passive.
\end{remark}

\begin{assumption}\label{As:PositiveDissipativity}
    In the networked system $\Sigma$, each subsystem $\Sigma_i$ is $X_i$-dissipative with $X_i^{11}>0, \forall i\in\mathbb{N}_N$, and similarly, each subsystem $\bar{\Sigma}_i$ is $\bar{X}_i$-dissipative with $\bar{X}_i^{11}>0, \forall i\in\N_{\bar{N}}$.
\end{assumption}

\begin{remark}
    According to Rm. \ref{Rm:X-DissipativityVersions}, As. \ref{As:PositiveDissipativity} holds if a subsystem $\Sigma_i,i\in\N_N$ is either: (i) L2G($\gamma_i$) or (ii) IF-OFP($\nu_i,\rho_i$) with $\nu_i<0$ (i.e., $L_2$-stable or non-passive). Since in passivity-based control, often the involved subsystems are non-passive (or can be treated as such), As. \ref{As:PositiveDissipativity} is also mild. 
\end{remark}


\begin{proposition}\label{synthesizeM}\cite{welikala2023non}
    Under As. \ref{As:NegativeDissipativity}-\ref{As:PositiveDissipativity}, the network system $\Sigma$ can be made \textbf{Y}-dissipative (from $w(t)$ to $z(t)$) by synthesizing the interconnection matrix $M$ \eqref{interconnectionMatrix} via solving the LMI problem:
\begin{equation}
\begin{aligned}
	&\mbox{Find: } 
	L_{uy}, L_{u\bar{y}}, L_{uw}, L_{\bar{u}y}, L_{\bar{u}\bar{y}}, L_{\bar{u}w}, M_{zy}, M_{z\bar{y}}, M_{zw}, \\
	&\mbox{Sub. to: } p_i \geq 0, \forall i\in\N_N, \ \ 
	\bar{p}_l \geq 0, \forall l\in\N_{\bar{N}},\ \text{and} \  \eqref{NSC4YEID},
\end{aligned}
\end{equation}
with
\scriptsize
$\bm{M_{uy} & M_{u\bar{y}} & M_{uw} \\ M_{\bar{u}y} & M_{\bar{u}\bar{y}} & M_{\bar{u}w}} = 
\bm{\textbf{X}_p^{11} & \0 \\ \0 & \bar{\textbf{X}}_{\bar{p}}^{11}}^{-1} \hspace{-1mm} \bm{L_{uy} & L_{u\bar{y}} & L_{uw} \\ L_{\bar{u}y} & L_{\bar{u}\bar{y}} & L_{\bar{u}w}}$.
\normalsize
where $\textbf{X}_p^{kl} \triangleq \diag(\{p_iX_i^{kl}:i\in\N_N\}), \forall k,l\in\N_2$, $\textbf{X}^{12} \triangleq \diag((X_i^{11})^{-1}X_i^{12}:i\in\N_N)$, and 
$\textbf{X}^{21} \triangleq (\textbf{X}^{12})^\T$ (terms $\bar{\textbf{X}}_{\bar{p}}^{kl}$, $\bar{\textbf{X}}^{12}$ and 
$\bar{\textbf{X}}^{21}$ have analogous definitions).
\end{proposition}


\begin{figure*}[!hb]
\vspace{-5mm}
\centering
\hrulefill
\begin{equation}\label{NSC4YEID}
\scriptsize
 \bm{
		\textbf{X}_p^{11} & \0 & \0 & L_{uy} & L_{u\bar{y}} & L_{uw} \\
		\0 & \bar{\textbf{X}}_{\bar{p}}^{11} & \0 & L_{\bar{u}y} & L_{\bar{u}\bar{y}} & L_{\bar{u}w}\\
		\0 & \0 & -\textbf{Y}^{22} & -\textbf{Y}^{22} M_{zy} & -\textbf{Y}^{22} M_{z\bar{y}} & \textbf{Y}^{22} M_{zw}\\
		L_{uy}^\T & L_{\bar{u}y}^\T & - M_{zy}^\T\textbf{Y}^{22} & -L_{uy}^\T\textbf{X}^{12}-\textbf{X}^{21}L_{uy}-\textbf{X}_p^{22} & -\textbf{X}^{21}L_{u\bar{y}}-L_{\bar{u}y}^\T \bar{\textbf{X}}^{12} & -\textbf{X}^{21}L_{uw} + M_{zy}^\T \textbf{Y}^{21} \\
		L_{u\bar{y}}^\T & L_{\bar{u}\bar{y}}^\T & - M_{z\bar{y}}^\T\textbf{Y}^{22} & -L_{u\bar{y}}^\T\textbf{X}^{12}-\bar{\textbf{X}}^{21}L_{\bar{u}y} & 		-(L_{\bar{u}\bar{y}}^\T \bar{\textbf{X}}^{12} + \bar{\textbf{X}}^{21}L_{\bar{u}\bar{y}}+\bar{\textbf{X}}_{\bar{p}}^{22}) & -\bar{\textbf{X}}^{21} L_{\bar{u}w} + M_{z\bar{y}}^\T \textbf{Y}^{21} \\ 
		L_{uw}^\T & L_{\bar{u}w}^\T & -M_{zw}^\T \textbf{Y}^{22}& -L_{uw}^\T\textbf{X}^{12}+\textbf{Y}^{12}M_{zy} & -L_{\bar{u}w}^\T\bar{\textbf{X}}^{12}+ \textbf{Y}^{12} M_{z\bar{y}} & M_{zw}^\T\textbf{Y}^{21} + \textbf{Y}^{12}M_{zw} + \textbf{Y}^{11}
	}\normalsize > 0
\end{equation}
\end{figure*}

Before concluding this section, we recall three linear algebraic results that will be useful in the sequel.

\begin{lemma}\label{Lm:Schur_comp}
\textbf{(Schur Complement)} For matrices $P > 0, Q$ and $R$, the following statements are equivalent:
\begin{subequations}
\begin{align}
1)\ &\begin{bmatrix} P & Q \\ Q^\T & R \end{bmatrix} \geq 0, \label{Lm:Schur1}\\
2)\ &P \geq 0, R-Q^\T P^{-1}Q \geq 0, \label{Lm:Schur2}\\
3)\ &R \geq 0, P-QR^{-1}Q^\T \geq 0 \label{Lm:Schur3}.
\end{align}
\end{subequations}
\end{lemma}

\begin{proof}
We first establish the equivalence of \eqref{Lm:Schur1} and \eqref{Lm:Schur2}. Let $z = \begin{bmatrix} z_1 & z_2 \end{bmatrix}^\T$ be any non-zero vector. Then:
\begin{equation}\label{Schur4}
z^\T\begin{bmatrix} P & Q \\ Q^\T & R \end{bmatrix}z = z_1^\T Pz_1 + z_1^\T Qz_2 + z_2^\T Q^\T z_1 + z_2^\T Rz_2
\end{equation}
Let $y = z_2 + P^{-1}Qz_1$. Then \eqref{Schur4} can be rewritten as:
\begin{equation}
z_1^\T Pz_1 - z_1^\T QP^{-1}Qz_1 + y^\T(R-Q^\T P^{-1}Q)y
\end{equation}
For the first part $z_1^\T P z_1 - z_1^\T QP^{-1}Qz_1$, we can observe that when $P > 0$:

\begin{equation}
\begin{aligned}\label{Schur5}
&z_1^\T P z_1 - z_1^\T QP^{-1}Qz_1 = z_1^\T(P - QP^{-1}Q)z_1 =\\
&z_1^\T P^{1/2}(\I - P^{-1/2}QP^{-1}Q^\T P^{-1/2})P^{1/2}z_1
\end{aligned}
\end{equation}

Since $P^{-1/2}QP^{-1}Q^\T P^{-1/2}$ is positive semidefinite, its eigenvalues are non-negative. Therefore, $(\I - P^{-1/2}QP^{-1}Q^\T P^{-1/2})$ has eigenvalues less than or equal to 1, making this term positive semidefinite. Consequently, $z_1^\T P z_1 - z_1^\T QP^{-1}Qz_1 \geq 0$. Thus, the entire expression \eqref{Schur5} is non-negative if and only if $R - Q^TP^{-1}Q \geq 0$. This establishes the equivalence of \eqref{Lm:Schur1} and \eqref{Lm:Schur2}.

The equivalence of \eqref{Lm:Schur1} and \eqref{Lm:Schur3} can be established by following similar steps.
\end{proof}

\begin{lemma}\label{Lm:InvSchur}
For any $P > 0$ and a square matrix $Q$:  
$$Q^\T P^{-1}Q > Q^\T + Q - P.$$ 
\end{lemma}
\begin{proof}
For any arbitrary matrix $S$, since $P>0$: we have $(S-\I)^\T P (S-I) > 0$, which can be simplifies to  
$$S^\T P S - PS - S^\T P + P > 0.$$
Now, we can obtain the required result by applying the change of variables $S = P^{-1}Q$ and re-arranging the terms.
\end{proof}

\begin{corollary}\label{Co:Lm:InvSchur}
For any $P \in \R^{n\times n}$ and $Q\in \R^{n\times m}$ such that $P>0$ and $n<m$: 
$$
Q^\T P^{-1} Q > 
\bar{Q}^\T + \bar{Q} - \bm{P & \0_{n\times (m-n)} \\ \0_{(m-n)\times n} & \I_{m-n}}.$$ 
where $\bar{Q}^\T \triangleq \bm{Q^\T & \0_{m\times(m-n)}}$.
\end{corollary}
\begin{proof}
The proof is complete by observing
$$Q^\T P^{-1} Q = \bar{Q}^\T \bm{P^{-1} & \0_{n\times (m-n)} \\ \0_{(m-n)\times n} & \I_{m-n}} \bar{Q}$$ and applying Lm. \ref{Lm:InvSchur} for the expression in the right-hand side.  
\end{proof}

\begin{lemma}\label{Lm:Woodbury}
For an invertible $R \in \mathbb{R}^{n \times n}$ and $\rho \in \R_{>0}$:
$$
(R + \rho I)^{-1} = R^{-1} - \rho R^{-1} \left( I + \rho R^{-1} \right)^{-1} R^{-1}.
$$
\end{lemma}
\begin{proof}
The result follows directly from applying the well-known Woodbury Matrix Identity:
$$
(R + UV^T)^{-1} = R^{-1} - R^{-1}U \left( I + V^T R^{-1}U \right)^{-1} V^T R^{-1}
$$
with the choices $U = \sqrt{\rho} I$ and $V = \sqrt{\rho} I$.
\end{proof}

\section{Problem Formulation}\label{problemformulation}
This section presents the dynamic modeling of the DC MG, which consists of multiple DGs, loads, and transmission lines. Specifically, our modeling approach is motivated by \cite{nahata}, which highlights the role and impact of communication and physical topologies in DC MGs.

\subsection{DC MG Physical Interconnection Topology}
The physical interconnection topology of a DC MG is modeled as a directed connected graph $\mathcal{G}^p =(\mathcal{V},\mathcal{E})$ where $\mathcal{V} \triangleq \mathcal{D} \cup \mathcal{L}$ is bipartite: $\mathcal{D} \triangleq \{\Sigma_i^{DG}, i\in\N_N\}$ (DGs) and $\mathcal{L} \triangleq \{\Sigma_l^{line}, l\in\N_L\}$ (lines). The DGs are interconnected with each other through transmission lines. The interface between each DG and the DC MG is through a point of common coupling (PCC). For simplicity, the loads are assumed to be connected to the DG terminals at the respective PCCs \cite{dorfler2012kron}. Indeed loads can be moved to PCCs using Kron reduction even if they are located elsewhere \cite{dorfler2012kron}.


To represent the DC MG's physical topology, we use its bi-adjacency matrix $\mathcal{A} \triangleq  \scriptsize 
\begin{bmatrix}
\0 & \mathcal{B} \\
\mathcal{B}^\T & \0
\end{bmatrix},
\normalsize$  
where $\mathcal{B}\in\R^{N \times L}$ is the incident matrix of the DG network (where nodes are just the DGs and edges are just the transmission lines). Note that $\mathcal{B}$ is also known as the ``bi-adjacency'' matrix of $\mathcal{G}^p$ that describes the connectivity between its two types of nodes. In particular, $\mathcal{B} \triangleq [\mathcal{B}_{il}]_{i \in \N_N, l \in \N_L}$ with
$\mathcal{B}_{il} \triangleq \mb{1}_{\{l\in\mathcal{E}_i^+\}} - \mb{1}_{\{l\in\mathcal{E}_i^-\}},$ where $\mathcal{E}_i^+$ and $\mathcal{E}_i^-$ represent the out- and in-neighbors of $\Sigma_i^{DG}$, respectively.


\subsection{Dynamic Model of a Distributed Generator (DG)}
Each DG consists of a DC voltage source, a voltage source converter (VSC), and some RLC components. Each DG $\Sigma_i^{DG},i\in\N_N$ supplies power to a specific ZIP load at its PCC (denoted $\text{PCC}_i$). Additionally, it interconnects with other DG units via transmission lines $\{\Sigma_l^{line}:l \in \mathcal{E}_i\}$. Figure \ref{DCMG} illustrates the schematic diagram of $\Sigma_i^{DG}$, including the local ZIP load, a connected transmission line, and the steady state, local, and distributed global controllers.

By applying Kirchhoff's Current Law (KCL) and Kirchhoff's Voltage Law (KVL) at $\text{PCC}_i$ on the DG side, we get the following equations for $\Sigma_i^{DG},i\in\N_N$:
\begin{equation}
\begin{aligned}\label{DGEQ}
\Sigma_i^{DG}:
\begin{cases}
    C_{ti}\frac{dV_i}{dt} &= I_{ti} - I_{Li}(V_i) - I_i + w_{vi}, \\
L_{ti}\frac{dI_{ti}}{dt} &= -V_i - R_{ti}I_{ti} + V_{ti} + w_{ci},
\end{cases}
\end{aligned}
\end{equation}
where the parameters $R_{ti}$, $L_{ti}$, and $C_{ti}$ 
represent the internal resistance, internal inductance, and filter capacitance of $\Sigma_i^{DG}$, respectively. The state variables are selected as $V_i$ and $I_{ti}$, where $V_i$ is the $\text{PCC}_i$ voltage and $I_{ti}$ is the internal current. Moreover, $V_{ti}$ is the input command signal applied to the VSC, $I_{Li}(V_i)$ is the total current drawn by the ZIP load, and $I_i$ is the total current injected to the DC MG by $\Sigma_i^{DG}$. We have also included $w_{vi}$ and $w_{ci}$ terms in \eqref{DGEQ} to represent unknown disturbances (assumed bounded and zero mean) resulting from external effects or modeling imperfections. 

Note that $V_{ti}$, $I_{Li}(V_i)$, and $I_i$ terms in \eqref{DGEQ} are respectively determined by the controllers, ZIP loads, and transmission lines at $\Sigma_i^{DG}$. Their details will be provided in the sequel. Let us begin with the total line current $I_i$, which is given by 
\begin{equation}
\label{Eq:DGCurrentNetOut}
I_i
=\sum_{l\in\mathcal{E}_i}\mathcal{B}_{il}I_l,
\end{equation}
where $I_l, l\in\mathcal{E}_i$ are line currents.

\subsection{Dynamic Model of a Transmission Line}
As shown in Fig. \ref{DCMG}, the power line $\Sigma_l^{line}$ can be represented as an RL circuit with resistance $R_l$ and inductance $L_l$. By applying KVL to $\Sigma_l^{line}$, we obtain:
\begin{equation}\label{line}
    \Sigma_l^{line}: 
    \begin{cases}  
        L_l\frac{dI_l}{dt}=-R_lI_l+\Bar{u}_l + \bar{w}_l,
    \end{cases}
\end{equation}
where $I_l$ is the line current (i.e., the state), $\Bar{u}_l=V_i-V_j=\sum_{i\in \mathcal{E}_l}\mathcal{B}_{il}V_i$ is the voltage differential (i.e., the line input), and $\bar{w}_l(t)$ represents the unknown disturbance (assumed bounded and zero mean) that affects the line dynamics.


\begin{figure}
    \centering
    \includegraphics[width=1\columnwidth]{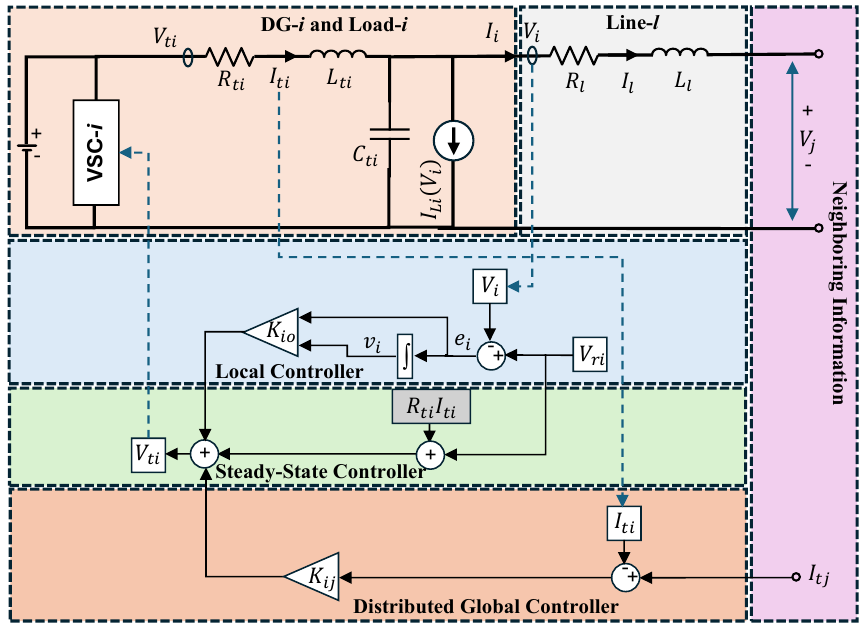}
    \caption{The electrical schematic of DG-$i$, load-$i$, $i\in\N_N$, local controller, distributed global controller, and line-$l$, $l\in\N_L$.}
    \label{DCMG}
\end{figure}

\subsection{Dynamic Model of a ZIP Load} 
Recall that $I_{Li}(V_i)$ in \eqref{DGEQ} (see also Fig. \ref{DCMG}) is the total current drawn by the load at $\Sigma_i^{DG}, i\in\N_N$. As the load is assumed to be a generic ``ZIP'' load, $I_{Li}(V_i)$ takes the form:
\begin{equation}\label{Eq:LoadModel}
I_{Li}(V_i) = I_{Li}^Z(V_i) + I_{Li}^I(V_i) + I_{Li}^P(V_i).
\end{equation}

Here, the ZIP load's components are: 
(i) a constant impedance load: $I_{Li}^{Z}(V_i)=Y_{Li}V_i$, where  $Y_{Li}=1/R_{Li}$ is the conductance of this load component;
(ii) a constant current load: $I_{Li}^{I}(V_i)=\bar{I}_{Li}$, where $\bar{I}_{Li}$ is the current demand of this load component; and
(iii) a constant power load (CPL): $I_{Li}^{P}(V_i)=V_i^{-1}P_{Li}$, where $P_{Li}$ represents the power demand of this load component. 

As opposed to $I_{Li}^{Z}(V_i)$ and $I_{Li}^{I}(V_i)$ (that take an affine linear form in the chosen state variables), the CPL $I_{Li}^{P}(V_i)$ introduces significant stability challenges due to its inherent negative impedance characteristic. This can be observed by examining the small-signal impedance of the CPL:
\begin{equation}
    Z_{CPL} = \frac{\partial V_i}{\partial I_{Li}^P} = \frac{\partial V_i}{\partial (P_{Li}/V_i)} = -\frac{V_i^2}{P_{Li}} < 0.
\end{equation}
This negative impedance characteristic creates a destabilizing effect in the DC MG, as it tends to amplify voltage perturbations rather than dampen them \cite{emadi2006constant}. When a small voltage drop occurs, the CPL draws more current to maintain constant power, further reducing the voltage and potentially leading to voltage collapse if not properly controlled.

The nonlinear nature of CPLs also introduces complexities for the control design. In particular, the nonlinear term $I_{Li}^P(V_i) = V_i^{-1}P_{Li}$ appears in the voltage dynamics (not in current dynamics) channel in \eqref{DGEQ}, and hence cannot be directly canceled using state feedback linearization techniques.  
Consequently, this nonlinearity must be carefully accounted for to ensure system stability and robustness, as often CPLs constitute a significant portion of the total ZIP load. 

As we will see in the subsequent sections, the proposed control framework exploits a key structural property of this nonlinearity,  namely, its sector boundedness, to address these stability and robustness concerns posed by CPLs.

\section{Proposed Hierarchical Control Architecture}\label{Sec:Controller}

The primary control objective of the DC MG is to ensure that the $\text{PCC}_i$ voltage $V_i$ at each $\Sigma_i^{DG}, i\in\N_N $ closely follows a specified reference voltage $V_{ri}$ while maintaining a proportional current sharing among DGs (with respect to their power ratings). In the proposed control architecture, these control objectives are achieved through the complementary action of local and distributed controllers. The local controller at each $\Sigma_i^{DG}$ is a PI controller responsible for the voltage regulation task. On the other hand, the distributed global controller at each $\Sigma_i^{DG}$ is a consensus-based controller responsible for ensuring proper current sharing among DGs.

\subsection{Local Voltage Regulating Controller}
At each $\Sigma_i^{DG}, i\in\N_N$, for its PCC$_i$ voltage $V_i(t)$ to effectively track the assigned reference voltage $V_{ri}(t)$, it is imperative to ensure that the tracking error $e_i(t)\triangleq V_i(t)-V_{ri}(t)$ converges to zero, i.e. $\lim_{t \to \infty} (V_i(t) - V_{ri}) = 0$. To this end, motivated by \cite{tucci2017}, we first include each $\Sigma_i^{DG}, i\in\N_N$ with an integrator state $v_i$ defined as $v_i(t) \triangleq \int_0^t (V_i(\tau) - V_{ri})d\tau$) (see also Fig. \ref{DCMG}) that follows the dynamics 
\begin{equation}\label{error}
    \frac{dv_i(t)}{dt}=e_i(t)=V_i(t)-V_{ri}.
\end{equation}
 Then, $\Sigma_i^{DG}$ is equipped with a local state feedback controller
\begin{equation}\label{Controller}
  u_{iL}(t)\triangleq  k_{i0}^P (V_i-V_{ri}) + k_{i0}^I v_i(t) = K_{i0}x_i(t) - k_{i0}^P V_{ri},
\end{equation}
where 
\begin{equation}\label{Eq:DGstate}
x_i \triangleq \begin{bmatrix}
    V_i &  I_{ti} & v_i
\end{bmatrix}^\top,
\end{equation}
denotes the augmented state (henceforth referred to as the state) of $\Sigma_i^{DG}$ and $K_{i0}\triangleq \begin{bmatrix}
    k_{i0}^P & 0 & k_{i0}^I
\end{bmatrix}\in\mathbb{R}^{1\times3}$ where $K_{i0}$ is the local controller gain.

\subsection{Distributed Global Controller}
We implement distributed global controllers at each DG, and task them with maintaining a proportional current sharing among the DGs. In particular, their objective is to ensure:
\begin{equation}\label{Eq:PropCurrSharing}
\frac{I_{ti}(t)}{P_{ni}} = \frac{I_{tj}(t)}{P_{nj}} = I_s, \quad \forall i,j\in\mathbb{N}_N,
\end{equation}
where $P_{ni}$ and $P_{nj}$ represent the power ratings of DGs  $\Sigma_i^{DG}$ and $\Sigma_j^{DG}$ respectively, and $I_s$ represents the common current sharing ratio that emerges from balancing the total load demand among DGs according to their power ratings.

To address the current sharing, as shown in Fig. \ref{DCMG}, we employ a consensus-based distributed controller
\begin{equation}\label{ControllerG}
    u_{iG}(t)\triangleq  \sum_{j\in\bar{\mathcal{F}}_i^-} k_{ij}^I\left(\frac{I_{ti}(t)}{P_{ni}} - \frac{I_{tj}(t)}{P_{nj}}\right),
\end{equation} 
where each $k_{ij}^I\in\mathbb{R}$ is a consensus controller gain.


Note that we denote the communication topology as a directed graph $\mathcal{G}^c =(\mathcal{D},\mathcal{F})$ where $\mathcal{D}\triangleq\{\Sigma_i^{DG}, i\in\N_N\}$ and $\mathcal{F}$ represents the set of communication links among DGs. The notations $\mathcal{F}_i^+$ and $\mathcal{F}_i^-$ (see \eqref{ControllerG}) are defined as the communication-wise out- and in-neighbors, respectively.

Finally, the overall control input $u_i(t)$ applied to the VSC of $\Sigma_i^{DG}$ (i.e., as $V_{ti}(t)$ in \eqref{DGEQ}) can be expressed as
\begin{equation}\label{controlinput}
    u_i(t) \triangleq V_{ti}(t) =  u_{iS} + u_{iL}(t) + u_{iG}(t),
\end{equation}
where $u_{iL}$ is given by \eqref{Controller}, $u_{iG}$ is given by \eqref{ControllerG} and $u_{iS}$ represents the steady-state control input.

As we will see in the sequel, steady-state control input $u_{iS}$ in \eqref{controlinput} also plays a crucial role in achieving the desired equilibrium point of the DC MG. In particular, this steady-state component ensures that the system maintains its operating point that satisfies both voltage regulation and current sharing objectives. The specific structure and properties of $u_{iS}$ will be characterized through our stability analysis presented in Sec. \ref{Sec:Equ_Analysis}. 



\subsection{Closed-Loop Dynamics of the DC MG}

By combining \eqref{DGEQ} and \eqref{error}, the overall dynamics of $\Sigma_i^{DG}, i\in\N_N$ can be written as
\begin{subequations}\label{statespacemodel}
\begin{align}
       \frac{dV_i}{dt}&=\frac{1}{C_{ti}}I_{ti}-\frac{1}{C_{ti}}I_{Li}(V_i)-\frac{1}{C_{ti}}I_i+\frac{1}{C_{ti}}w_{vi}(t), \label{Eq:ss:voltages}\\
        \frac{dI_{ti}}{dt}&=-\frac{1}{L_{ti}}V_i-\frac{R_{ti}}{L_{ti}}I_{ti}+\frac{1}{L_{ti}}u_i + \frac{1}{L_{ti}}w_{ci}(t), \label{Eq:ss:currents}\\
        \frac{dv_i}{dt}&=V_i-V_{ri} \label{Eq:ss:ints},
\end{align}
\end{subequations}
where the terms  $I_i$, $I_{Li}(V_i)$, and $u_i$ can be substituted from \eqref{Eq:DGCurrentNetOut}, \eqref{Eq:LoadModel}, and \eqref{controlinput}, respectively. We can restate \eqref{statespacemodel} as
\begin{equation}
\label{Eq:DGCompact}
\dot{x}_i(t)= A_ix_i(t)+B_iu_i(t)+E_id_i(t)+\xi_i(t) + g_i(x_i(t)),
\end{equation}
where $x_i(t)$ is the DG state as defined in \eqref{Eq:DGstate}, $d_i(t)$ is the exogenous input (disturbance) defined as 
\begin{equation}
d_i(t) \triangleq 
\bar{w}_i + w_i(t),
\end{equation}
with 
$\bar{w}_i \triangleq \bm{
    -\Bar{I}_{Li}  & 0 & -V_{ri}
}^\T$ representing the fixed (mean) known disturbance and
$w_i(t) \triangleq \bm{w_{vi}(t) & w_{ci}(t) & 0}^\T$ representing the zero-mean unknown disturbance, $E_i \triangleq \diag(\bm{C_{ti}^{-1} & L_{ti}^{-1} & 1})$ is the disturbance input matrix, $\xi_i(t)$ is the transmission line coupling input defined as 
$\xi_i(t) \triangleq \begin{bmatrix}
    -C_{ti}^{-1}\sum_{l\in \mathcal{E}_i} \mathcal{B}_{il}I_l(t) & 0 & 0
\end{bmatrix}^\top$, $g_i(x_i(t))$ represents the nonlinear vector field due to the CPL defined as 
$$
g_i(x_i(t)) \triangleq C_{ti}^{-1}\begin{bmatrix} -\frac{P_{Li}}{V_i} & 0 & 0 \end{bmatrix}^\T,
$$
and $A_i$ and $B_i$ are system matrices respectively defined as 
\begin{equation}\label{Eq:DG_Matrix_definition}
A_i \triangleq 
\begin{bmatrix}
  -\frac{Y_{Li}}{C_{ti}} & \frac{1}{C_{ti}} & 0\\
-\frac{1}{L_{ti}} & -\frac{R_{ti}}{L_{ti}} & 0 \\
1 & 0 & 0
\end{bmatrix}\ \mbox{ and }\  
B_i \triangleq
\begin{bmatrix}
 0 \\ \frac{1}{L_{ti}} \\ 0
\end{bmatrix}.
\end{equation}

Similarly, using \eqref{line}, the state space representation of the transmission line $\Sigma_l^{Line}$ can be written in a compact form:
\begin{equation}\label{Eq:LineCompact}
    \dot{\bar{x}}_l(t) = \bar{A}_l\bar{x}_l(t) + \bar{B}_l\bar{u}_l(t) + \bar{E}_l\bar{w}_l(t),
 \end{equation}
where $\bar{x}_l \triangleq I_l$ is the transmission line state,  $\bar{E}_l \triangleq \bm{\frac{1}{L_{l}}}$ is the disturbance matrix, and $\bar{A}_l$ and $\Bar{B}_l$ are the system matrices respectively defined as
\begin{equation}\label{Eq:LineMatrices}
\bar{A}_l \triangleq 
\begin{bmatrix}
-\frac{R_l}{L_l}
\end{bmatrix}\ 
\mbox{ and }\ 
\Bar{B}_l \triangleq  \begin{bmatrix}
\frac{1}{L_l}
\end{bmatrix}.  
\end{equation}

\subsection{Networked System Model}\label{Networked System Model}

Let us define $u\triangleq[u_i]_{i\in\N_N}^\T$ and $\Bar{u}\triangleq[\Bar{u}_l]_{l\in\N_L}^\T$ respectively as vectorized control inputs of DGs and lines, 
$x\triangleq[x_i^\T]_{i\in\N_N}^\T$ and $\Bar{x}\triangleq[\Bar{x}_l]_{l\in\N_L}^\T$ respectively as the full states of DGs and lines, $w\triangleq[w_i^\T]_{i\in\N_N}^\T$ and $\bar{w}\triangleq[\bar{w}_l]_{l\in\N_L}^\T$ respectively as disturbance inputs of DGs and lines.

Using these notations, we can now represent the closed-loop DC MG as two sets of subsystems (i.e., DGs and lines) interconnected with disturbance inputs through a static interconnection matrix $M$ as shown in Fig. \ref{netwoked}. From comparing Fig. \ref{netwoked} with Fig. \ref{Networked}, it is clear that the DC MG takes a similar form to a standard networked system discussed in Sec. \ref{SubSec:NetworkedSystemsPreliminaries}.

\begin{figure}
    \centering
    \includegraphics[width=0.9\columnwidth]{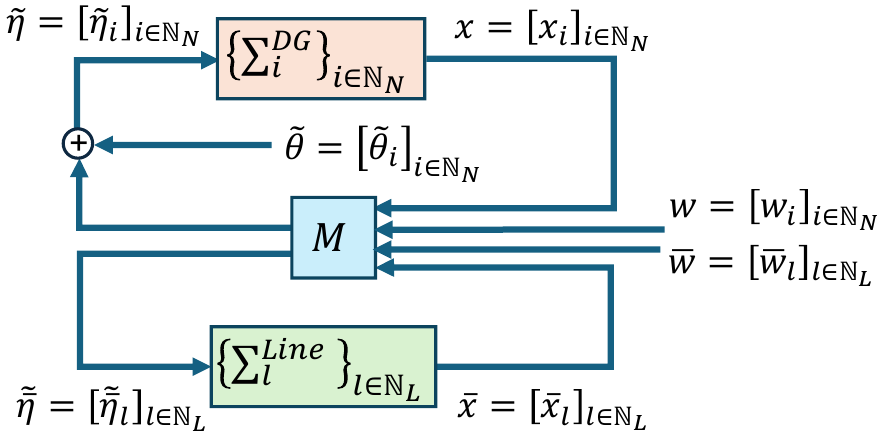}
    \caption{DC MG dynamics as a networked system configuration.}
    \label{netwoked}
\end{figure}

To identify the specific structure of the interconnection matrix $M$ in Fig. \ref{netwoked} (i.e., for DC MG), we need to closely observe how the dynamics of DGs and lines are interconnected and how they are coupled with disturbance inputs.

To this end, we first use \eqref{Eq:DGCompact} and \eqref{controlinput} to state the closed-loop dynamics of $\Sigma_i^{DG}$ as:
\begin{equation}\label{closedloopdynamic}
    \dot{x}_i(t) = (A_i+B_iK_{i0})x_i(t)+\tilde{\eta}_i(t),
\end{equation}
where $\tilde{\eta}_i$ is defined as 
\begin{equation}
\label{Eq:DGClosedLoopDynamics_varphi}
\tilde{\eta}_i(t) \triangleq E_iw_i(t) +\sum_{l\in\mathcal{E}_i}\Bar{C}_{il}\Bar{x}_l(t)+\sum_{j\in\bar{\mathcal{F}}_i^-}K_{ij}x_j(t) + \theta_i,
\end{equation}
with
$\Bar{C}_{il} \triangleq -C_{ti}^{-1}\begin{bmatrix}
\mathcal{B}_{il} & 0 & 0
\end{bmatrix}^\top$, 
$$\theta_i \triangleq  E_i\bar{w}_i + B_iu_{iS} - B_ik_{i0}^P V_{ri},$$ and $K_{ij}$ represents the distributed consensus controller gain matrix (for current sharing objective), that takes the form:
\begin{equation}\label{k_ij}
    K_{ii} \triangleq
    \frac{1}{L_{ti}}
    \bm{
        0 & 0 & 0\\
        0 & \frac{\sum_{j\in\mathcal{F}_i^-} k_{ij}^I}{P_{ni}} & 0 \\
        0 & 0 & 0
    },\
        K_{ij} \triangleq
    \frac{1}{L_{ti}}
    \bm{
        0 & 0 & 0 \\
        0 & \frac{-k_{ij}^I}{P_{nj}} & 0 \\
        0 & 0 & 0    
    }.
\end{equation}

From \eqref{k_ij}, observe that only the (2,2)-th element in each block $K_{ij}$ is non-zero. Let $K_I \in \R^{N \times N}$ denote the matrix that contains only these (2,2) block entries, i.e., $K_I = [K_{ij}^{2,2}]_{i,j\in\N_N}$. The controller gain matrix satisfies the weighted Laplacian property:
\begin{equation}
    K_I P_n \textbf{1}_N = 0,
\end{equation}
where $P_n = \diag(\bm{P_{ni}}_{i\in\N_N})$ and $\textbf{1}_N \in \R^N$ is the vector of ones. This ensures that the distributed control vanishes when proportional current sharing is achieved among all DGs.

By vectorizing \eqref{Eq:DGClosedLoopDynamics_varphi} over all $i\in\N_N$, we get 
\begin{equation}\label{Eq:DGClosedLoopInputVector}
    \tilde{\eta} \triangleq  Ew+\Bar{C}\Bar{x}+Kx+\theta,
\end{equation}
where $\tilde{\eta} \triangleq [\tilde{\eta}_i^\T]_{i\in \N_N}^\T$ represents the effective input vector to the DGs (see Fig. \ref{netwoked}), $E \triangleq \diag([E_i]_{i\in\N_N})$ represents the disturbance matrix of DGs, $\Bar{C}\triangleq[\Bar{C}_{il}]_{i\in\mathbb{N}_N,l\in\mathbb{N}_L}$, $K \triangleq [K_{ij}]_{i,j\in\mathbb{N}_N}$, and $\theta\triangleq [\theta_i]_{i\in\N_N}^\T$ represents a constant (time-invariant) input vector applied to the DGs.

\begin{remark}
The block matrices $K$ and $\Bar{C}$ in \eqref{Eq:DGClosedLoopInputVector} are indicative of the communication and physical topologies of the DC MG, respectively. In particular, the $(i,j)$\tsup{th} block in $K$, i.e., $K_{ij}$, indicates a communication link from $\Sigma_j^{DG}$ to $\Sigma_i^{DG}$. Similarly, $(i,l)$\tsup{th} block in $\bar{C}$, i.e., $\bar{C}_{i,l}$ indicates a physical link between $\Sigma_i^{DG}$ and $\Sigma_l^{Line}$.
\end{remark}


Similarly to DGs, using \eqref{Eq:LineCompact}, we state the closed-loop dynamics of $\Sigma_l^{Line}$ as  
\begin{equation}
    \dot{\bar{x}}_l(t) = \bar{A}_l\bar{x}_l(t) + \tilde{\bar{\eta}}_l(t),
\end{equation}
where $\tilde{\bar{\eta}}_l(t)$ is defined as 
\begin{equation}\label{Eq:linecontroller}
    \tilde{\Bar{\eta}}_l(t) \triangleq \sum_{i\in\mathcal{E}_l}C_{il}x_i(t) + \bar{E}_l\bar{w}_l(t),
\end{equation}
with $C_{il}\triangleq \begin{bmatrix}
    \mathcal{B}_{il} & 0 & 0
\end{bmatrix}$ (note also that $C_{il} = -C_{ti}\bar{C}_{il}^\T$). By vectorizing \eqref{Eq:linecontroller} over all $l\in\N_L$, we get
\begin{equation}\label{ubar}
    \tilde{\Bar{\eta}} \triangleq Cx + \bar{E}\bar{w},
\end{equation}
where $\tilde{\bar{\eta}} \triangleq [\tilde{\bar{\eta}}_l]_{l\in\N_L}^\T$ represents the effective input vector to the lines, $C\triangleq[C_{il}]_{l\in\mathbb{N}_L,i\in\mathbb{N}_N}$ (note also that $C = - \bar{C}^\T C_t$ where $C_t \triangleq \diag([C_{ti}\I_{3}]_{i\in\N_N})$), $\bar{E} \triangleq \diag([\bar{E}_l]_{l\in\N_L})$, and $\bar{w} \triangleq [\bar{w}_l]_{l\in\N_L}$.





Finally, using \eqref{Eq:DGClosedLoopInputVector} and \eqref{ubar}, we can identify the interconnection relationship:  
\begin{equation}\nonumber
\bm{\tilde{\eta}^\T & \tilde{\bar{\eta}}^\T
}^\T
= M
\bm{x^\T & \bar{x}^\T & w^\T & \bar{w}^\T}^\T,
\end{equation}
where the interconnection matrix $M$ takes the form:
\begin{equation}\label{Eq:MMatrix}
M \triangleq 
\begin{bmatrix}
    K & \Bar{C} & E & \0 \\
    C & \0 & \0 & \bar{E}\\
\end{bmatrix}.
\end{equation}
When the physical topology $\mathcal{G}^p$ is predefined, so are the block matrices $\Bar{C}$ and $C$ (recall $C = -\bar{C}^\T C_t$). This leaves only the block matrix $K$ inside the block matrix $M$ as a tunable quantity to optimize the desired properties of the closed-loop DC MG system. Note that synthesizing $K$ simultaneously determines the distributed global controllers \eqref{ControllerG} and the communication topology $\mathcal{G}^c$. In the following two sections, we formulate this networked system's error dynamics (around a desired operating point)  and provide a systematic dissipativity-based approach to synthesize this block matrix $K$ to enforce dissipativity of the closed-loop error dynamics (from disturbance inputs to a given performance output).

\section{Nonlinear Networked Error Dynamics}\label{Sec:ControlDesign}

This section establishes the mathematical foundation for analyzing the stability and robustness of the DC MG with ZIP loads. We begin with rigorous equilibrium point analysis to characterize steady-state behavior, deriving necessary conditions for simultaneous voltage regulation and current sharing. We then develop nonlinear error dynamics that explicitly account for CPL characteristics, defining error variables relative to the desired operating/equilibrium point. Next, we transform the complex DC MG error dynamics into a standard networked system structure with clearly defined performance outputs and disturbance inputs by constructing a complete state-space representation capturing DG and transmission line error dynamics. This formulation creates the foundation for our subsequent dissipativity-based control design technique, which will be discussed in the next section.

\subsection{Equilibrium Point Analysis of the DC MG}\label{Sec:Equ_Analysis}

In this section, we analyze the equilibrium conditions of the DC MG to establish mathematical relationships between system parameters and steady-state behavior. This analysis is crucial for identifying the necessary conditions for simultaneously achieving voltage regulation and proportional current sharing. We pay particular attention to the impact of CPL components, which introduce nonlinear dynamics and potentially lead to instability in the DC MG system.

\begin{lemma}\label{Lm:equilibrium}
Assuming all zero mean unknown disturbance components to be zero, i.e., $w_i(t)=\0, \forall i\in\N_N$ and $\bar{w}_l(t)=0,\forall l\in\N_L$, for a given reference voltage vector $V_r$, under a fixed control input $u(t) = u_E \triangleq[u_{iE}]_{i\in\N_N}$ defined as
\begin{equation}
    u_E \triangleq [\I + R_t(\mathcal{B}R^{-1}\mathcal{B}^\T + Y_L)]V_r + R_t(\bar{I}_L + \diag(V_r)^{-1}P_L), 
\end{equation}
there exists an equilibrium point for the DC MG characterized by reference voltage vector $V_r \triangleq [V_{ri}]_{i\in\N_N}^\T$, constant current load vector $\bar{I}_L \triangleq [\bar{I}_{Li}]_{i\in\N_N}$, and CPL vector $P_L \triangleq [P_{Li}]_{i\in\N_N}^\T$, given by:
\begin{equation}\label{Eq:Equilibrium}
\begin{aligned}
V_E &= V_{r},\\
I_{tE} &= (\mathcal{B}R^{-1}\mathcal{B}^\T + Y_L)V_r + \bar{I}_L + \diag(V_r)^{-1}P_L,\\
\bar{I}_E &= R^{-1}\mathcal{B}^\T V_r.
\end{aligned}
\end{equation}
where we define the state equilibrium vectors $V_E\triangleq[V_{iE}]_{i\in\N_N}$,
$I_{tE}\triangleq[I_{tiE}]_{i\in\N_N}$,
$\bar{I}_E\triangleq[\bar{I}_{lE}]_{l\in\N_L}$, 
and the system parameters 
$Y_L\triangleq\diag([Y_{Li}]_{i\in\N_N})$,
$R_t\triangleq\diag([R_{ti}]_{i\in\N_N})$,
and $R\triangleq\diag([R_l]_{l\in\mathbb{N}_L})$.
\end{lemma}

\begin{proof}
The equilibrium state of the closed-loop dynamic $\Sigma_i^{DG}$ \eqref{Eq:DGCompact} satisfies:
\begin{equation}
\label{Eq:EqDGCompact}
A_ix_{iE}(t)+B_iu_{iE}(t)+E_id_{iE}(t)+\xi_{iE}(t) = 0,
\end{equation}
where $x_{iE} \triangleq \bm{V_{iE} & I_{tiE} & v_{iE}}^\T$ represents the equilibrium state components of DG, and $w_{iE}\triangleq \bar{w}_i$ and $\xi_{iE}$ represent the equilibrium values of disturbance and interconnection terms, respectively. Thus, we get
\begin{equation}
\begin{aligned}
\begin{bmatrix} -\frac{Y_{Li}}{C_{ti}} & -\frac{1}{C_{ti}} & 0\\ -\frac{1}{L_{ti}} & -\frac{R_{ti}}{L_{ti}} & 0 \\ 1 & 0 & 0\end{bmatrix} \begin{bmatrix} V_{iE} \\ I_{tiE} \\ v_{iE} \end{bmatrix} + \begin{bmatrix} 0 \\ \frac{1}{L_{ti}} \\ 0 \end{bmatrix}u_{iE} + E_i\bar{w}_i + \xi_{iE} = 0.
\end{aligned}
\end{equation}
From respective rows of this matrix equation, we get:
\begin{align}
\label{Eq:steadystatevoltage}
&-\frac{Y_{Li}}{C_{ti}}V_{iE} + \frac{1}{C_{ti}}I_{tiE} - \frac{1}{C_{ti}}\bar{I}_{Li} - \frac{P_{Li}}{C_{ti}V_{iE}} - \frac{1}{C_{ti}}\sum_{l\in \mathcal{E}i} \mathcal{B}_{il}\bar{I}_{lE} = 0,\\
\label{Eq:steadystatecurrent}
&-\frac{1}{L_{ti}}V_{iE} - \frac{R_{ti}}{L_{ti}}I_{tiE} + \frac{1}{L_{ti}}u_{iE} = 0, \\
\label{Eq:steadystateintegrator}
&V_{iE} - V_{ri} = 0.
\end{align}
From the last two equations above, we can obtain
\begin{align}
 \label{Eq:voltageEquil}
    &V_{iE} = V_{ri},\\  
    \label{Eq:controlEquil}
    &u_{iE} = V_{iE} + R_{ti}I_{tiE}.
\end{align}
To simplify the first equation further, we require to know an expression for $\bar{I}_{lE}$. 

Note that the equilibrium state of the $\Sigma_l^{Line}$ \eqref{Eq:LineCompact} satisfies:
\begin{equation}\label{Eq:EqLineCompact}
    \bar{A}_l\bar{x}_{lE}(t) + \bar{B}_l\bar{u}_{lE} + \bar{E}_l\bar{w}_{lE}(t) = 0,
 \end{equation}
where $\bar{x}_{lE} \triangleq \bar{I}_{lE}$ represents the equilibrium state of line. The $\bar{u}_{lE}\triangleq \sum_{i\in \mathcal{E}_l}\mathcal{B}_{il}V_{iE}$ and $\bar{w}_{lE} \triangleq 0$, respectively, represent the equilibrium values of control input and disturbance of lines. Therefore, we get:
\begin{equation}\label{Eq:steadystateLine}
    -\frac{R_l}{L_l}\bar{I}_{lE} + \frac{1}{L_l}\sum_{i\in \mathcal{E}_l}\mathcal{B}_{il}V_{iE} = 0,
\end{equation}
leading to
\begin{equation}
\label{Eq:LineEquil}
    \bar{I}_{lE} = \frac{1}{R_l}\sum_{i\in \mathcal{E}_l}\mathcal{B}_{il}V_{iE} = \frac{1}{R_l}\sum_{j\in \mathcal{E}_l}\mathcal{B}_{jl}V_{jE},
\end{equation}
which can be applied in \eqref{Eq:steadystatevoltage} (together with \eqref{Eq:voltageEquil}) to obtain
\begin{equation}\label{Eq:currentEquil}
    -Y_{Li}V_{ri} + I_{tiE} - \sum_{l\in\mathcal{E}_i}\mathcal{B}_{il}\bigg(\frac{1}{R_l}\sum_{j\in \mathcal{E}_l}\mathcal{B}_{jl}V_{rj}\bigg)- \bar{I}_{Li} - \frac{P_{Li}}{V_{ri}} = 0.
\end{equation}

We next vectorize these equilibrium conditions. Note that, then the control equilibrium equation \eqref{Eq:controlEquil} becomes:
\begin{equation}\label{Eq:controlEquilVec}
    u_E = V_r + R_tI_{tE}.
\end{equation}
Vectorizing the voltage dynamics equation \eqref{Eq:currentEquil}, we get:
\begin{equation}\nonumber
    -Y_L V_r + I_{tE} - \mathcal{B}R^{-1}\mathcal{B}^\T V_r - \bar{I}_L - \diag(V_r)^{-1}P_L = 0,
\end{equation}
leading to
\begin{equation}\nonumber
    I_{tE} = (\mathcal{B}R^{-1}\mathcal{B}^\T + Y_L)V_r + \bar{I}_L + \diag(V_r)^{-1}P_L.
\end{equation}
Therefore, the vectorized control equilibrium equation can be expressed as:
\begin{equation}
\begin{aligned}\nonumber
    u_E &= V_r + R_t I_{tE} \\
    &= V_r + R_t((\mathcal{B}R^{-1}\mathcal{B}^\T + Y_L)V_r + \bar{I}_L + \diag(V_r)^{-1}P_L)\\
    &= [\I + R_t(\mathcal{B}R^{-1}\mathcal{B}^\T + Y_L)]V_r + R_t(\bar{I}_L + \diag(V_r)^{-1}P_L).
    \end{aligned}
\end{equation}
For the equilibrium line currents, by vectorizing \eqref{Eq:LineEquil}, we get 
\begin{equation}\nonumber
    \bar{I}_E = R^{-1}\mathcal{B}^\T V_r.
\end{equation}
This completes the proof, as we have derived all the required equilibrium conditions. 



\end{proof}

\begin{remark}
The uniqueness is mathematically guaranteed because the diagonal matrices $R$, $R_t$, and $Y_L$ have strictly positive elements, making them positive definite, while the incidence matrix $\mathcal{B}$ maintains full rank by virtue of the connected network topology. Consequently, the coefficient matrix $(\mathcal{B}R^{-1}\mathcal{B}^\T + Y_L)$ in \eqref{Eq:Equilibrium} is invertible, which ensures a unique one-to-one mapping from any given reference voltage vector $V_r$ to all equilibrium variables under specified loading conditions.
\end{remark}

\begin{remark}\label{Lm:currentsharing}
At the equilibrium, we require the condition for proportional current sharing among DGs to meet (i.e., \eqref{Eq:PropCurrSharing}), and thus, we require
\begin{equation}\label{Eq:currentsharing}
\frac{I_{tiE}}{P_{ni}} = I_s \Longleftrightarrow I_{tiE} = P_{ni}I_s, \ \forall i\in\mathbb{N}_N, 
\end{equation}
which can be expressed in vectorized form as $I_{tE} = P_{n} \mathbf{1}_N I_s$, where $P_n\triangleq\diag([P_{ni}]_{i\in\mathbb{N}_N})$. Using this requirement in \eqref{Eq:controlEquilVec}, we get $u_E = V_r +  R_t P_{n} \mathbf{1}_N I_s,$ i.e., 
$$
u_{iE} = V_{ri} +  R_{ti} P_{ni} I_s, \quad \forall i\in\N_N. 
$$
Therefore, to achieve this particular control equilibrium (which satisfies both voltage regulation and current sharing objectives), we need to select our steady-state control input in \eqref{controlinput} as:
\begin{equation}\label{Eq:Rm:SteadyStateControl}
    u_{iS} = V_{ri} + R_{ti}P_{ni}I_s, \quad \forall i\in\N_N. 
\end{equation}

This is because at the equilibrium point, local control $u_{iL}$ \eqref{Controller} and distributed global control $u_{iG}$ \eqref{ControllerG} components are, by definition, zero for any $i\in\N_N$.
\end{remark}

In conclusion, using Lm. \ref{Lm:equilibrium} and Rm. \ref{Lm:currentsharing}, for the equilibrium of DC MG to satisfy the voltage regulation and current sharing conditions, we require:
\begin{equation}
    \begin{aligned}
        u_E &= [\I + R_t(\mathcal{B}R^{-1}\mathcal{B}^\T + Y_L)]V_r + R_t(\bar{I}_L + \diag(V_r)^{-1}P_L)  \\
        &= V_r +  R_t P_{n} \mathbf{1}_N I_s = u_S\\
        V_E &= V_r, \\
        I_{tE} &= (\mathcal{B}R^{-1}\mathcal{B}^\T + Y_L)V_r + \bar{I}_L + \diag(V_r)^{-1}P_L\\
        &= P_n \textbf{1}_N I_s,\\
        \bar{I}_E &= R^{-1} \mathcal{B}^\T V_r.
    \end{aligned}
\end{equation}
The following theorem formalizes the optimization problem derived from Lm. \ref{Lm:equilibrium} and Rm. \ref{Lm:currentsharing}, for the selection of $V_r$ and $I_s$, that ensures the existence of an equilibrium state that satisfies voltage regulation and current sharing conditions while also respecting reference voltage limits $V_{\min}$ and $V_{\max}$ and the current sharing coefficient $I_s \in [0,1]$. A formal proof is omitted as the result follows directly from the equilibrium relationships and the remarks above.







\begin{theorem}
To ensure the existence of an equilibrium point that satisfies the voltage regulation and current sharing objectives, the reference voltages $V_r$ and current sharing coefficient $I_s$ should be a feasible solution in the optimization problem:
\begin{equation}
\begin{aligned}
\min_{V_r,I_s}\ &\alpha_V\Vert V_r - \bar{V}_r\Vert^2 + \alpha_I I_s \\
\mbox{Sub. to:}\ & V_{\min} \leq V_r \leq V_{\max},
\quad 0 \leq I_s \leq 1,
\end{aligned}
\end{equation}
$$P_n \mathbf{1}_N I_s - (\mathcal{B}R^{-1}\mathcal{B}^\T + Y_L)V_r = \bar{I}_L + \diag(V_r)^{-1}P_L,$$
where $\bar{V}_r$ is a desired reference voltage value, and $\alpha_V > 0$ and $\alpha_I > 0$ are two normalizing cost coefficients.
\end{theorem}

It is worth noting that the above optimization problem becomes an LMI problem (convex) when the CPL is omitted (i.e., when $P_L=\0$). 
Overall, this formulation ensures proper system operation through multiple aspects. The equality constraint guarantees that the current sharing objective is achieved across all DG units. The reference voltage bounds maintain system operation within safe and efficient limits through the inequality constraints on $V_r$. Furthermore, the constraint on $I_s$ ensures that the current sharing coefficient remains properly normalized for practical implementations. As state earlier, the nonlinear term $\diag(V_r)^{-1}P_L$ introduces additional complexity in determining a feasible set of reference voltages and a current sharing coefficient that simultaneously satisfy voltage regulation and current sharing objectives.




\subsection{Nonlinear Error Dynamics with CPL}
The network system representation described in Sec. \ref{Networked System Model} can be simplified by considering the error dynamics around the identified equilibrium point in Lm. \ref{Lm:equilibrium}. As we will see in the sequel, the resulting error dynamics can be seen as a networked system (called the networked error system) comprised of DG error subsystems, line error subsystems, external disturbance inputs, and performance outputs.


We first define error variables that capture deviations from the identified equilibrium:
\begin{subequations}
\begin{align}
    \tilde{V}_i &= V_i - V_{iE} = V_i - V_{ri}, \label{Eq:errorvariable1}\\
    \tilde{I}_{ti} &= I_{ti} - I_{tiE} = {I_{ti} - P_{ni}I_s}, \label{Eq:errorvariable2}\\
     \tilde{v}_i &= v_i - v_{iE}, \label{Eq:errorvariable3}\\
    \tilde{I}_l &= I_{l} - \bar{I}_{lE} = I_l - \frac{1}{R_l}\sum_{i\in\E_l}\mathcal{B}_{il}V_{ri}.
\end{align}
\end{subequations}
Now, considering the dynamics \eqref{Eq:ss:voltages}-\eqref{Eq:ss:ints}, equilibrium point established in Lm. \ref{Lm:equilibrium}, and the proposed a hierarchical control strategy $u_i(t)$ \eqref{controlinput}, the error dynamics can then be derived as follows.

The voltage error dynamics can be derived using \eqref{Eq:ss:voltages} and \eqref{Eq:errorvariable1} as:
\begin{equation}\label{Eq:currenterrordynamic}
\begin{aligned}
    \dot{\tilde{V}}_i = &-\frac{Y_{Li}}{C_{ti}}(\tilde{V}_i+V_{ri}) + \frac{1}{C_{ti}}(\tilde{I}_{ti} +P_{ni}I_s) - \frac{1}{C_{ti}}\bar{I}_{Li}\\
    &- \frac{1}{C_{ti}}\sum_{l\in \mathcal{E}i} \mathcal{B}_{il}(\tilde{I}_l + \frac{1}{R_l}\sum_{j\in\E_l}\mathcal{B}_{jl}V_{rj})\\
    &- \frac{1}{C_{ti}}(\tilde{V}_i + V_{ri})^{-1}P_{Li}+ \frac{1}{C_{ti}}w_{vi}\\
    \equiv& \frac{1}{C_{ti}}\Big(\phi_V + \psi_V +g_i(\tilde{V}_i)\Big) + \frac{1}{C_{ti}}w_{vi},
\end{aligned}
\end{equation}
where 
\begin{subequations}
\begin{align}
   \phi_V &\triangleq -Y_{Li}\tilde{V}_i + \tilde{I}_{ti} - \sum_{l\in \mathcal{E}i} \mathcal{B}_{il}\tilde{I}_l,\\
    \psi_V &\triangleq -Y_{Li}V_{ri}+P_{ni}I_s - \bar{I}_{Li} - \sum_{l\in \mathcal{E}i} \frac{\mathcal{B}_{il}}{R_l}\sum_{j\in\E_l}\mathcal{B}_{jl}V_{rj} - \frac{V_{ri}}{P_{Li}}, \label{Eqvoltageerror}\\
    g_i(\tilde{V}_i) &\triangleq V_{ri}^{-1}P_{Li} - (\tilde{V}_i + V_{ri})^{-1}P_{Li}.
\end{align}
\end{subequations}
The current error dynamics can be obtained using \eqref{Eq:ss:currents} and \eqref{Eq:errorvariable2} as:
\begin{equation}\label{Eq:voltageerrordynamic}
\begin{aligned}
    \dot{\tilde{I}}_{ti} &= -\frac{1}{L_{ti}}(\tilde{V}_i+V_{ri}) - \frac{R_{ti}}{L_{ti}}(\tilde{I}_{ti} +P_{ni}I_s) \\
    &+ \frac{1}{L_{ti}}(u_{iS}+k_{i0}^P\tilde{V}_i+k_{io}^I\tilde{v}_i+\sum_{j\in\bar{\mathcal{F}}_i^-} k_{ij}(\frac{\tilde{I}_{ti}}{P_{ni}} - \frac{\tilde{I}_{tj}}{P_{nj}}))\\
    &+ \frac{1}{L_{ti}}w_{ci},\\
    &\equiv \frac{1}{L_{ti}}\Big(\phi_I + \psi_I\Big) + \frac{1}{L_{ti}}w_{ci},
\end{aligned}
\end{equation}
where
\begin{subequations}
\begin{align}
    \phi_I \triangleq&-\tilde{V}_i - R_{ti}\tilde{I}_{ti} + k_{io}^p\tilde{V}_i + k_{io}^I\tilde{v}_i
    +\sum_{j\in\bar{\mathcal{F}}_i^-} k_{ij}(\frac{\tilde{I}_{ti}}{P_{ni}} - \frac{\tilde{I}_{tj}}{P_{nj}})),\\
     \psi_I \triangleq& -V_{ri}-R_{ti}P_{ni}I_s + u_{iS}. \label{Eqcurrenteerror}
\end{align}
\end{subequations}
The integral error dynamics can be achieved by using \eqref{Eq:ss:ints} and \eqref{Eq:errorvariable3} as:
\begin{equation}\label{Eq:integralerrordynamic}
    \dot{\tilde{v}}_i = \tilde{V}_i.
\end{equation}

It is worth noting that, as a consequence of the equilibrium analysis and the steady state control input selection (see \eqref{Eq:currentsharing} and \eqref{Eq:Rm:SteadyStateControl}), the terms $\Psi_V$ \eqref{Eqvoltageerror} and $\Psi_I$ \eqref{Eqcurrenteerror} are canceled. Therefore, for each DG error subsystem $\tilde{\Sigma}_i^{DG}, i\in\N_N$, we have an error state vector $\tilde{x}_i = \begin{bmatrix} \tilde{V}_i, \tilde{I}_{ti}, \tilde{v}_i \end{bmatrix}^\T$ with the dynamics: 
\begin{equation}\label{Eq:DG_error_dynamic}
    \dot{\tilde{x}}_i = (A_i + B_i K_{i0})\tilde{x}_i  + g_i(\tilde{x}_i) + \tilde{u}_i, 
\end{equation}
where $\tilde{u}_i$ represents the interconnection input combining the effects of both line currents and other DG states, defined as
\begin{equation}
\tilde{u}_i \triangleq u_i + E_iw_i = \sum_{l\in\mathcal{E}_i}\Bar{C}_{il}\tilde{\Bar{x}}_l +  \sum_{j\in\bar{\mathcal{F}}_i^-}K_{ij}\tilde{x}_j + E_iw_i,
\end{equation}
$g_i(\tilde{x}_i)$ is the nonlinear vector due to the CPL, defined as
\begin{equation}\label{nonlinear_vector}
    g_i(\tilde{x}_i) = \bm{\frac{1}{C_{ti}}\big(V_{ri}^{-1}P_{Li} - (\tilde{V}_i + V_{ri})^{-1}P_{Li}\big) \\ 0 \\0
    },
\end{equation}
and the system and disturbance matrices, $A_i$ and $E_i$, respectively, are same as before.


Following similar steps, we can obtain the dynamics of the transmission line error subsystem $\tilde{\Sigma}_l^{Line}, l\in\N_L$ as:
\begin{equation}\label{Eq:Line_error_dynamic}
    \dot{\tilde{\bar{x}}}_l = \bar{A}_l\tilde{\bar{x}}_l + \tilde{\bar{u}}_l,  
\end{equation}
where $\tilde{\bar{u}}_l$ represents the line interconnection input influenced by DG voltages and disturbances:
\begin{equation}
    \tilde{\bar{u}}_l = \bar{u}_l + \bar{E}_l \bar{w}_l = \sum_{i\in\mathcal{E}_l}B_{il}\tilde{V}_i  + \bar{E}_l \bar{w}_l.
\end{equation}

To ensure robust stability (dissipativity) of this networked error system, we define local performance outputs as follows. For each DG error subsystem $\tilde{\Sigma}_i^{DG}, i\in\N_N$, we define the performance output as:
\begin{equation}
z_i(t) \triangleq H_i\tilde{x}_i(t),
\end{equation}
where $H_i$ can be selected as $H_i \triangleq \I, \forall i\in\N_N$ (not necessarily). 
Similarly, for each line error subsystem $\tilde{\Sigma}_l^{Line}, l\in\N_L$, we define the performance output as:
\begin{equation}
\bar{z}_l(t) \triangleq \bar{H}_l\tilde{\bar{x}}_l(t),
\end{equation}
where $\bar{H}_l$ can be selected as $\bar{H}_l \triangleq \I, \forall l\in\N_L$ (not necessarily). 

Upon vectorizing these performance outputs over all $i\in\N_N$ and $l\in\N_L$ (respectively), we obtain: 
\begin{equation}\label{z}
\begin{aligned}
     z \triangleq H\tilde{x}\quad \mbox{ and } \quad 
     \bar{z} \triangleq \bar{H}\tilde{\bar{x}}
\end{aligned}
\end{equation} 
where $H \triangleq \diag([H_i]_{i\in\N_N})$ and $\bar{H} \triangleq \diag([\bar{H}_l]_{l\in\N_L})$. This choice of performance output mapping provides a direct correspondence between error subsystem states and the performance outputs.

Defining $z \triangleq \bm{z_i^\T}^\T_{i\in\N_N}$ and $\bar{z} \triangleq \bm{\bar{z}_l^\T}^\T_{l\in\N_L}$, we consolidate the performance outputs and disturbance inputs as
\begin{equation}
    \begin{aligned}
        z_c \triangleq \bm{z^\T & \bar{z}^\T}^\T\quad \mbox{ and } \quad 
        w_c \triangleq \bm{w^\T & \bar{w}^\T}^\T.
    \end{aligned}
\end{equation}
The consolidated disturbance vector $w_c$ affects the networked error dynamics, particularly the DG error subsystems and the line error subsystems, respectively, through the consolidated disturbance matrices $E_c$ and $\bar{E}_c$, defined as:
\begin{equation}
    \begin{aligned}
        E_c \triangleq \bm{E & \0}\quad \mbox{ and }\quad 
        \bar{E}_c \triangleq \bm{\0 & \bar{E}}.
    \end{aligned}
\end{equation}

The zero blocks in the $E_c$ and $\bar{E}_c$ indicate that line disturbances do not directly affect DG error subsystem inputs and vice versa. Analogously, the dependence of consolidated performance outputs on the networked error system states can be described using consolidated performance matrices 
\begin{equation}
    \begin{aligned}
        H_c \triangleq \bm{H \\ \0}\quad \mbox{ and }\quad 
        \bar{H}_c \triangleq \bm{\0 \\ \bar{H}}.
    \end{aligned}
\end{equation}

With these definitions and the derived error subsystem dynamics \eqref{Eq:DG_error_dynamic} and \eqref{Eq:Line_error_dynamic}, it is easy to see that the closed-loop error dynamics of the DC MG can be modeled as a networked error system as shown in Fig: \ref{Fig.DissNetError}. In there, the interconnection relationship between the error subsystems, disturbance inputs and performance outputs is described by:
\begin{equation}
    \begin{bmatrix}
        \tilde{u} & \tilde{\bar{u}} & z_c
    \end{bmatrix}^\T = M \begin{bmatrix}
        \tilde{x} & \tilde{\bar{x}} & w_c
    \end{bmatrix}^\T
\end{equation}
where the interconnection matrix $M$ takes the form:
\begin{equation}\label{Eq:NetErrSysMMat}
    M \triangleq \bm{M_{\tilde{u}x} & M_{\tilde{u}\bar{x}} & M_{\tilde{u}w_c}\\
    M_{\tilde{\bar{u}}x} &  M_{\tilde{\bar{u}}\bar{x}} &  M_{\tilde{\bar{u}}w_c}\\
    M_{z_cx} &  M_{z_c\bar{x}} &  M_{z_cw_c}
        } \equiv 
        \bm{K & \bar{C} & E_c\\
        C & \0 & \bar{E}_c\\
        H_{c} & \bar{H}_{c} & \0
        }.
\end{equation}

\begin{figure}
    \centering
    \includegraphics[width=0.99\columnwidth]{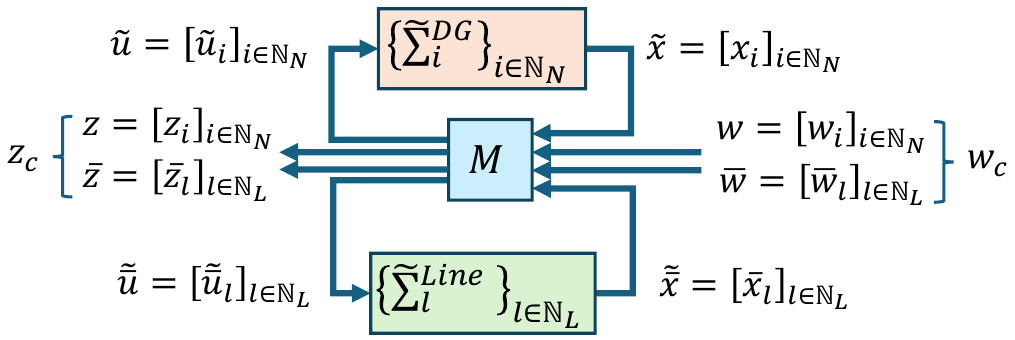}
    \caption{DC MG  error dynamics as a networked system with disturbance inputs and performance outputs.}
    \label{Fig.DissNetError}
\end{figure}

\section{Dissipativity-Based Co-Design Framework for Controllers and Communication Topology}\label{Passivity-based Control}

In this section, we first introduce the global control and topology co-design problem for the DC MG using the derived networked error dynamics representation in the previous section (see Fig. \ref{Fig.DissNetError}). As this co-design problem exploits dissipativity properties of the involved subsystem error dynamics, we next identify necessary conditions for subsystem dissipativity properties. Subsequently, we embed these necessary conditions in the local dissipativating controller design problems implemented at the subsystems. Finally, the overall control design process is summarized.


\subsection{Error Subsystem Dissipativity Properties}
Consider the DG error subsystem $\tilde{\Sigma}_i^{DG},i\in\mathbb{N}_N$ \eqref{Eq:DG_error_dynamic} to be $X_i$-dissipative with
\begin{equation}\label{Eq:XEID_DG}
    X_i=\begin{bmatrix}
        X_i^{11} & X_i^{12} \\ X_i^{21} & X_i^{22}
    \end{bmatrix}\triangleq
    \begin{bmatrix}
        -\nu_i\mathbf{I} & \frac{1}{2}\mathbf{I} \\ \frac{1}{2}\mathbf{I} & -\rho_i\mathbf{I}
    \end{bmatrix},
\end{equation}
where $\rho_i$ and $\nu_i$ are the passivity indices of $\tilde{\Sigma}_i^{DG}$. In other words, consider $\tilde{\Sigma}_i^{DG}, i\in\mathbb{N}_N$ to be IF-OFP($\nu_i,\rho_i$). It is worth noting that the IF-OFP($\nu_i,\rho_i$) property assumed here will be enforced through local controller design in Sec. \ref{Sec:Local_Synth} (Th. \ref{Th:LocalControllerDesign}), where the passivity indices $\nu_i$ and $\rho_i$ are determined alongside the controller gains $K_{i0}$.

Similarly, consider the line error subsystem $\tilde{\Sigma}_l^{Line},l\in\mathbb{N}_L$ \eqref{Eq:Line_error_dynamic} to be $\bar{X}_l$-dissipative with
\begin{equation}\label{Eq:XEID_Line}
    \bar{X}_l=\begin{bmatrix}
        \bar{X}_l^{11} & \bar{X}_l^{12} \\ \bar{X}_l^{21} & \bar{X}_l^{22}
    \end{bmatrix}\triangleq
    \begin{bmatrix}
        -\bar{\nu}_l\mathbf{I} & \frac{1}{2}\mathbf{I} \\ \frac{1}{2}\mathbf{I} & -\bar{\rho}_l\mathbf{I}
    \end{bmatrix},
\end{equation}
where $\bar{\rho}_l$ and $\bar{\nu}_l$ are the passivity indices of $\tilde{\Sigma}_l^{Line}$. Regarding these passivity indices, we can provide the following lemma.

\begin{lemma}\label{Lm:LineDissipativityStep}
For each line $\tilde{\Sigma}_l^{Line}, l\in\N_L$ \eqref{Eq:LineCompact}, its passivity indices $\bar{\nu}_l$,\,$\bar{\rho}_l$ assumed in (\ref{Eq:XEID_Line}) are such that the LMI problem: 
\begin{equation}\label{Eq:Lm:LineDissipativityStep1}
\begin{aligned}
\mbox{Find: }\ &\bar{P}_l, \bar{\nu}_l, \bar{\rho}_l\\
\mbox{Sub. to:}\ &\bar{P}_l > 0, \ 
\begin{bmatrix}
    \frac{2\bar{P}_lR_l}{L_l}-\bar{\rho}_l & -\frac{\bar{P}_l}{L_l}+\frac{1}{2}\\
    \star & -\bar{\nu}_l
\end{bmatrix}
\normalsize
\geq0, 
\end{aligned}
\end{equation}
is feasible. The maximum feasible values for $\bar{\nu}_l$ and $\bar{\rho}_l$ respectively are $\bar{\nu}_l^{\max}=0$ and $\bar{\rho}_l^{\max}=R_l$, when $\bar{P}_l =  \frac{L_l}{2}$. 
\end{lemma}
Note that the feasibility of problem \eqref{Eq:Lm:LineDissipativityStep1} is guaranteed when the line resistance $R_l>0$ and inductance $L_l>0$, which hold for real-world lines.

\begin{proof}
For each $\tilde{\Sigma}^{Line}_l$, $l \in \mathcal{N}_L$ described by \eqref{Eq:Line_error_dynamic}, we need to ensure it is $\bar{X}_l$-dissipative with the passivity indices defined in \eqref{Eq:XEID_Line}. For this, we can apply Prop. \ref{Prop:linear_X-EID} with the given system matrices from \eqref{Eq:LineMatrices} and the specified dissipativity supply rate form in \eqref{Eq:XEID_Line}, leading to the LMI condition:
\begin{equation}
\begin{bmatrix} 2\bar{P}_l\frac{R_l}{L_l} - \bar{\rho}_l & -\bar{P}_l\frac{1}{L_l} + \frac{1}{2} \\ * & -\bar{\nu}_l \end{bmatrix} \geq 0
\end{equation}
Using Lm. \ref{Lm:Schur_comp}, this is positive semidefinite if and only if:
\begin{equation} 
\begin{aligned} 
&1)\ \bar{\nu}_l \leq 0,\\
&2)\ 2\bar{P}_l\frac{R_l}{L_l} - \bar{\rho}_l - \frac{(-\bar{P}_l\frac{1}{L_l} + \frac{1}{2})^2}{-\bar{\nu}_l} \geq 0.
\end{aligned}
\end{equation}

To maximize the passivity indices, we set $\bar{\nu}_l = 0$. With this choice, we need:
\begin{equation}
   \begin{aligned}
&1)\ 2\bar{P}_l\frac{R_l}{L_l} - \bar{\rho}_l \geq 0,\\
&2)\ -\bar{P}_l\frac{1}{L_l} + \frac{1}{2} = 0, 
   \end{aligned} 
\end{equation}
which gives $\bar{P}_l = \frac{L_l}{2}$.

Substituting $\bar{P}_l = \frac{L_l}{2}$ into condition 1:
$R_l - \bar{\rho}_l \geq 0$, implying $\bar{\rho}_l \leq R_l$. Therefore, the maximum feasible values are $\bar{\nu}_l^{\max} = 0$ and $\bar{\rho}_l^{\max} = R_l$, when $\bar{P}_l = \frac{L_l}{2}$.
\end{proof}

While we could identify the conditions required of the passivity indices of the line error dynamics \eqref{Eq:Line_error_dynamic}, achieving a similar feat for DG error dynamics is not straightforward due to the involved CPL nonlinearities (see \eqref{Eq:DG_error_dynamic}). This challenge is addressed in the following subsection.

\subsection{Sector-Bounded Characterization of CPL Nonlinearities}

The destabilizing negative impedance characteristics of CPLs pose significant challenges to controller design, requiring a specialized mathematical treatment for CPL nonlinearities. In the following discussion, we present a systematic approach to incorporate CPL nonlinearities into our dissipativity-based control framework using sector-boundedness concepts. 

First, for notational convenience (although it is a slight abuse of notation), we denote the first (and only non-zero) component of the CPL nonlinearity $g_i(\tilde{x}_i) \in \R^3$ \eqref{nonlinear_vector} as $g_i(\tilde{V}_i) \in \R$ (with a little abuse of notation), where 
\begin{equation} \label{nonlinear_scalar}
g_i(\tilde{V}_i) \triangleq \frac{P_{Li}}{C_{ti}}\left(\frac{1}{V_{ri}} - \frac{1}{\tilde{V}_i + V_{ri}}\right).
\end{equation}
The following lemma establishes the sector boundedness of this $g_i(\tilde{V}_i)$.

\begin{lemma}\label{Lm:sector_bounded}
For the CPL nonlinearity $g_i(\tilde{V}_i)$ \eqref{nonlinear_scalar}, there exist constants $\alpha_i, \beta_i \in \mathbb{R}_{\geq 0}$ such that:
\begin{equation} \label{Eq:Lm:sector_bounded1}
\alpha_i \leq \frac{g_i(\tilde{V}_i)}{\tilde{V}_i} \leq \beta_i, \quad \forall \tilde{V}_i \in [\tilde{V}_i^{\min}, \tilde{V}_i^{\max}] \backslash \{0\},
\end{equation}
where $\alpha_i \triangleq \frac{P_{Li}}{V_{\max}^2}$ and $\beta_i \triangleq \frac{P_{Li}}{V_{\min}^2}$, $V_{\min}$ and $V_{\max}$ represents the assumed operating voltage range, and $\tilde{V}_i^{\min} \triangleq  V_{\min} - V_{ri}$ and $\tilde{V}_i^{\max} \triangleq  V_{\max} - V_{ri}$ denotes the operating voltage error range. Moreover, \eqref{Eq:Lm:sector_bounded1} implies that the CPL nonlinearity $g_i(\tilde{V}_i)$  \eqref{nonlinear_scalar} satisfies the quadratic constraint:
\begin{equation} \label{Eq:Lm:sector_bounded2}  
\begin{bmatrix} \tilde{V}_i \\ g_i(\tilde{V}_i) \end{bmatrix}^\T \begin{bmatrix} -\alpha_i\beta_i & \frac{\alpha_i+\beta_i}{2} \\ \frac{\alpha_i+\beta_i}{2} & -1 \end{bmatrix} \begin{bmatrix} \tilde{V}_i \\ g_i(\tilde{V}_i) \end{bmatrix} \geq 0,
\end{equation}
for all $\tilde{V}_i \in [\tilde{V}_i^{\min}, \tilde{V}_i^{\max}]\backslash \{0\}$.
\end{lemma}

\begin{proof}
The derivative of the CPL nonlinearity $g_i(\tilde{V}_i)$  \eqref{nonlinear_scalar} with respect to $\tilde{V}_i$ is:
\begin{equation}
\frac{\partial g_i(\tilde{V}_i)}{\partial \tilde{V}_i} = \frac{P_{Li}}{C_{ti}(\tilde{V}_i + V_{ri})^2}
\end{equation}
Since $\tilde{V}_i + V_{ri} \in [V_{\min}, V_{\max}]$ (or, equivalently, $\tilde{V}_i\in [\tilde{V}_i^{\min}, \tilde{V}_i^{\max}]$) within the operating range, we have:
\begin{equation}  
\frac{P_{Li}}{C_{ti}V_{\max}^2} \leq \frac{\partial g_i(\tilde{V}_i)}{\partial \tilde{V}_i} \leq \frac{P_{Li}}{C_{ti}V_{\min}^2}
\end{equation}
By the mean value theorem, for any $\tilde{V}_i \in  [\tilde{V}_i^{\min}, \tilde{V}_i^{\max}] \backslash \{0\}$:
\begin{equation}
\frac{g_i(\tilde{V}_i)}{\tilde{V}_i} = \frac{\partial g_i(\xi)}{\partial \tilde{V}_i}
\end{equation}
for some $0 < \xi < \tilde{V}_i$. Therefore:
\begin{equation} 
\frac{P_{Li}}{C_{ti}V_{\max}^2} \leq \frac{g_i(\tilde{V}_i)}{\tilde{V}_i} \leq \frac{P_{Li}}{C_{ti}V_{\min}^2}
\end{equation}
Setting $\alpha_i = \frac{P_{Li}}{C_{ti}V_{\max}^2}$ and $\beta_i = \frac{P_{Li}}{C_{ti}V_{\min}^2}$, we get \eqref{Eq:Lm:sector_bounded1}. 

To derive the quadratic constraint \eqref{Eq:Lm:sector_bounded2}, we rewrite \eqref{Eq:Lm:sector_bounded1} as:
\begin{equation}  
\alpha_i\tilde{V}_i^2 \leq \tilde{V}_i g_i(\tilde{V}_i) \leq \beta_i\tilde{V}_i^2
\end{equation}
This gives the conditions:
\begin{equation} 
\tilde{V}_i(g_i(\tilde{V}_i) - \alpha_i\tilde{V}_i) \geq 0, \quad
\tilde{V}_i(g_i(\tilde{V}_i) - \beta_i\tilde{V}_i) \leq 0
\end{equation}
By multiplying these conditions, we can (uniquely) obtain  
\begin{equation}
    (g_i(\tilde{V}_i) - \alpha_i\tilde{V}_i)(g_i(\tilde{V}_i) - \beta_i\tilde{V}_i) \leq 0.
\end{equation}
Expanding this expression and rearranging it into a quadratic form in $\bm{\tilde{V}_i & g_i(\tilde{V}_i)}^\T$ we can obtain the condition \eqref{Eq:Lm:sector_bounded2}. 
\end{proof}

Next, we exploit Lm. \ref{Lm:sector_bounded} to analyze the dissipativity properties of the DG error dynamics  $\tilde{\Sigma}_i^{DG}, i\in\N_N$ \eqref{Eq:DG_error_dynamic} (from its input $\tilde{u}_i$ to output $\tilde{x}_i$). For this purpose, let us consider a quadratic storage function $\mathrm{V}_i(\tilde{x}_i) \triangleq \tilde{x}_i^\T P_i \tilde{x}_i$, where $P_i > 0$. Consequently:
\begin{equation}\label{Eq:RateOfStorage1}
\begin{aligned}
\dot{\mathrm{V}}_i(\tilde{x}_i) &= 2\tilde{x}_i^\T P_i(\hat{A}_i\tilde{x}_i + g_i(\tilde{x}_i) + u_i)\\
&= 2\tilde{x}_i^\T P_i\hat{A}_i\tilde{x}_i + 2\tilde{x}_i^\T P_ig_i(\tilde{x}_i) + 2\tilde{x}_i^\T P_iu_i
\end{aligned}
\end{equation}
where $\hat{A}_i \triangleq A_i + B_i K_{i0}$. Now the key challenge is to properly bound the term $2\tilde{x}_i^\T P_ig_i(\tilde{x}_i)$ using the structural properties of $g_i(\tilde{x}_i)$. This challenge is addressed in the following lemma using Lm. \ref{Lm:sector_bounded} and the well-known S-procedure technique \cite{petersen2000s}.

\begin{lemma}\label{Lm:S-Procedure}
Let $T_i \triangleq \begin{bmatrix} 1 & 0 & 0 \end{bmatrix}^\T$ and $\textbf{T}_i \triangleq \scriptsize \bm{T_i & \0 \\ \0 & T_i} \normalsize$. If (and only if) there exists $\lambda_i \in \R_{>0}$ and $R_i=R_i^\T \in \mathbb{R}^{3\times3}$ such that:
\begin{equation}\label{Eq:Lm:S-Procedure1}
\begin{bmatrix} R_i & -P_i \\ -P_i & 0 \end{bmatrix} - \lambda_i \textbf{T}_i \begin{bmatrix} -\alpha_i\beta_i & \frac{\alpha_i+\beta_i}{2} \\ \frac{\alpha_i+\beta_i}{2} & -1 \end{bmatrix} \textbf{T}_i^\T 
 \geq 0,
\end{equation}
then:
\begin{equation}\label{Eq:Lm:S-Procedure2}
2\tilde{x}_i^\T P_i g_i(\tilde{x}_i) \leq \tilde{x}_i^\T R_i \tilde{x}_i.
\end{equation}
\end{lemma}

\begin{proof}
From Lm. \ref{Lm:sector_bounded}, we have:
\begin{equation}
\begin{bmatrix} \tilde{V}_i \\ g_i(\tilde{V}_i) \end{bmatrix}^\T \begin{bmatrix} -\alpha_i\beta_i & \frac{\alpha_i+\beta_i}{2} \\ \frac{\alpha_i+\beta_i}{2} & -1 \end{bmatrix} \begin{bmatrix} \tilde{V}_i \\ g_i(\tilde{V}_i) \end{bmatrix} \geq 0.
\end{equation}
Since $\tilde{V}_i = T_i^\T \tilde{x}_i$ and $g_i(\tilde{V}_i) = T_i^\T g_i(\tilde{x}_i)$, we get:
\begin{equation}
\begin{bmatrix} T_i^\T \tilde{x}_i  \\ T_i^\T g_i(\tilde{x}_i) \end{bmatrix}^\T \begin{bmatrix} -\alpha_i\beta_i & \frac{\alpha_i+\beta_i}{2} \\ \frac{\alpha_i+\beta_i}{2} & -1 \end{bmatrix} \begin{bmatrix} T_i^\T\tilde{x}_i \\ T_i^\T g_i(\tilde{x}_i) \end{bmatrix} \geq 0,
\end{equation}
which is equivalent to 
\begin{equation}\label{Eq:Lm:S-ProcedureStep1}
\begin{bmatrix} \tilde{x}_i  \\ g_i(\tilde{x}_i) \end{bmatrix}^\T \textbf{T}_i \begin{bmatrix} -\alpha_i\beta_i & \frac{\alpha_i+\beta_i}{2} \\ \frac{\alpha_i+\beta_i}{2} & -1 \end{bmatrix} \textbf{T}_i^\T \begin{bmatrix} \tilde{x}_i \\  g_i(\tilde{x}_i) \end{bmatrix} \geq 0.
\end{equation}

On the other hand, the desired bound \eqref{Eq:Lm:S-Procedure2} $2\tilde{x}_i^\T P_i g_i(\tilde{x}_i) \leq \tilde{x}_i^\T R_i \tilde{x}_i$ can be restated as:
\begin{equation} \label{Eq:Lm:S-ProcedureStep2}
\begin{bmatrix} \tilde{x}_i \\ g_i(\tilde{x}_i) \end{bmatrix}^\T \begin{bmatrix} R_i & -P_i \\ -P_i^\T & 0 \end{bmatrix} \begin{bmatrix} \tilde{x}_i \\ g_i(\tilde{x}_i) \end{bmatrix}\geq 0.
\end{equation}

By the matrix S-procedure, \eqref{Eq:Lm:S-ProcedureStep1} implies \eqref{Eq:Lm:S-ProcedureStep2} if and only if there exists $\lambda_i > 0$ and $R_i$ such that \eqref{Eq:Lm:S-Procedure1} holds. 
\end{proof}

As we will see in the sequel, the above result, while it can be used for dissipativity analysis, is not directly applicable for dissipativating control design. This is due to the involved change of variables like $Q_i \triangleq P_i^{-1}$ that are used in formulating LMI problems for such control design taks. To address this preemptively, we provide the following lemma, which identifies a sufficient condition (i.e., slightly weaker than \eqref{Eq:Lm:S-Procedure1}) for the desired bound \eqref{Eq:Lm:S-Procedure2}. It is worth noting that this sufficient condition is only an LMI in design variables $P_i^{-1}, R_i^{-1}$ and $\lambda^{-1}$ (not in $P_i, R_i$  and $\lambda_i$ as in \eqref{Eq:Lm:S-Procedure1}). 

Before proceeding, for notational convenience (in manipulating the condition \eqref{Eq:Lm:S-Procedure1}), let us define 
\begin{equation}\label{Eq:SectorBoundMatrixTransformed}
\Theta \triangleq [\Theta_{kl}]_{k,l\in\N_2} \equiv 
\textbf{T}_i \begin{bmatrix} -\alpha_i\beta_i & \frac{\alpha_i+\beta_i}{2} \\ \frac{\alpha_i+\beta_i}{2} & -1 \end{bmatrix} \textbf{T}_i^\T.
\end{equation}
Consequently, we get 
$\Theta_{11} = -\alpha_i\beta_i T_i T_i^\T$, 
$\Theta_{12} = \frac{\alpha_i + \beta_i}{2} T_i T_i^\T$,
$\Theta_{21} = \Theta_{12}^\T$, and $\Theta_{22} = -T_i T_i^\T$.

\begin{lemma}\label{Lm:S-ProcedureModified}
If there exists $\lambda_i \in \R$ and $R_i=R_i^\T \in \mathbb{R}^{3\times3}$ such that $0 < \lambda_i \leq 1$, $R_i > 0$ and 
\begin{align} \nonumber
\Psi_i \triangleq 
\bm{\I & \0 & \0 \\ \0 & P_i^{-1}T_i & \0} 
+ \bm{\I &\ \0\\ \0 & T_i^\T P_i^{-1} \\ \0 & \0} 
-\bm{R_i^{-1} & \0 \\ \0 & \I } 
\\ \label{Eq:Co:S-Procedure1}
- \bm{\Theta_{11} & \Theta_{12}P_i^{-1}+ \frac{1}{\lambda_i}\I \\
P_i^{-1}\Theta_{21} + \frac{1}{\lambda_i}\I & \0}
 \geq 0,
\end{align}
then the desired bound \eqref{Eq:Lm:S-Procedure2} holds (the column/row of zeros included at the end of the first/second term of $\Psi_i$ serves as a padding to remedy the dimension mismatch caused by $T_i$).
\end{lemma}

\begin{proof}
Using Co. \ref{Co:Lm:InvSchur}, the first three terms in $\Psi_i$ can be upper bounded by the product 
\begin{align}\nonumber
\bm{\I & \0 & \0 \\ \0 & P_i^{-1}T_i & \0}
\bm{R_i & \0 \\ \0 & \I } 
\bm{\I &\ \0\\ \0 & T_i^\T P_i^{-1} \\ \0 & \0} \\
= \bm{R_i & \0 \\ \0 & P_i^{-1}T_i T_i^\T P_i^{-1}} = \bm{R_i & \0 \\ \0 & -P_i^{-1} \Theta_{22} P_i^{-1}},
\end{align}
where the last step is from the definition \eqref{Eq:SectorBoundMatrixTransformed} $\Theta_{22} = - T_i T_i^\T$, i.e., a unique consequence of sector bounded nature of the CPLs. Consequently, an upper bound for $\Psi_i$ can be written as 
\begin{align*}
\bm{R_i & \0 \\ \0 & -P_i^{-1} \Theta_{22} P_i^{-1}}
- \bm{\Theta_{11} & \Theta_{12}P_i^{-1}+ \frac{1}{\lambda_i}\I \\
P_i^{-1}\Theta_{21} + \frac{1}{\lambda_i}\I & \0}\\
= 
\bm{R_i & - \frac{1}{\lambda_i}\I \\ - \frac{1}{\lambda_i}\I & \0 }
- \bm{\Theta_{11} & \Theta_{12}P_i^{-1} \\
P_i^{-1}\Theta_{21} & P_i^{-1} \Theta_{22} P_i^{-1}}\\
\leq \bm{\frac{1}{\lambda_i} R_i & - \frac{1}{\lambda_i}\I \\ - \frac{1}{\lambda_i}\I & \0 }
- \bm{\Theta_{11} & \Theta_{12}P_i^{-1} \\
P_i^{-1}\Theta_{21} & P_i^{-1} \Theta_{22} P_i^{-1}}\\
= \frac{1}{\lambda_i}\bm{\I & \0 \\ \0 & P_i^{-1}}
\left( 
\bm{R_i & - P_i \\ - P_i & \0 }
- \lambda_i \Theta \right)\bm{\I & \0 \\ \0 & P_i^{-1}} 
\triangleq \bar{\Psi}_i
\end{align*}

Finally, combining the above result together with the facts that $\text{\eqref{Eq:Co:S-Procedure1}}$, $\bar{\Psi}_i \geq \Psi_i$, $P_i > 0$, $\lambda_i > 0$ and  \eqref{Eq:Lm:S-Procedure1} $\implies$ \eqref{Eq:Lm:S-Procedure2}, we can conclude that 
\begin{align*}
\text{\eqref{Eq:Co:S-Procedure1}} \iff \Psi_i > 0 \implies \bar{\Psi}_i \geq  0& \\
\iff \bm{R_i & - P_i \\ - P_i & \0 } - \lambda_i \Theta > 0& \iff \text{\eqref{Eq:Lm:S-Procedure1}} \implies \text{\eqref{Eq:Lm:S-Procedure2}}.    
\end{align*}
This completes the proof.
\end{proof}

\begin{theorem}\label{Th:Local_CPL}
Under the sector-bounded property \eqref{Lm:sector_bounded} established in Lm. \ref{Lm:sector_bounded} and the quadratic bound  \eqref{Eq:Lm:S-Procedure2} established in Lm. \ref{Lm:S-ProcedureModified}, the DG error dynamics $\tilde{\Sigma}_i: \tilde{u}_i \rightarrow \tilde{x}_i, i\in\N_N$ \eqref{Eq:DG_error_dynamic} can be made IF-OFP($\nu_i$, $\rho_i$) (as assumed in \eqref{Eq:XEID_DG}) via designing the local controller gains matrix $K_{i0}$ \eqref{Controller} using the LMI problem:
\begin{equation}\label{Eq:Th:Local_CPL} 
\begin{aligned}
&\mbox{Find: }\ \tilde{K}_{i0},\ \tilde{P}_i,\ \tilde{R}_i,\ \tilde{\lambda}_i,\ \nu_i,\ \tilde{\rho}_i,\\
&\mbox{Sub. to: }\ \tilde{P}_i > 0,\ \tilde{R}_i > 0,\ \tilde{\lambda}_i > 1\\
&\begin{bmatrix}
\tilde{\rho}_i\I + \tilde{R}_i & \tilde{R}_i & \0 & \0 \\
\tilde{R}_i & \tilde{R}_i & \tilde{P}_i & \0 \\
\0 & \tilde{P}_i & -\mathcal{H}(A_i\tilde{P}_i + B_i\tilde{K}_{i0}) & -\I + \frac{1}{2}\tilde{P}_i \\ 
\0 & \0 & -\I + \frac{1}{2}\tilde{P}_i & -\nu_i\I    
\end{bmatrix} > 0,\\
& \bm{\I & \0 & \0 \\ \0 & \tilde{P}_i T_i & \0} 
+ \bm{\I &\ \0\\ \0 & T_i^\T \tilde{P}_i \\ \0 & \0} 
-\bm{\tilde{R}_i & \0 \\ \0 & \I } \\
&- \bm{\Theta_{11} & \Theta_{12}\tilde{P}_i + \tilde{\lambda}_i\I \\
\tilde{P}_i\Theta_{21} + \tilde{\lambda}_i\I & \0} 
\geq 0,    
\end{aligned}
\end{equation}
where $K_{i0} \triangleq \tilde{K}_{i0}\tilde{P}_i^{-1}$. 
\end{theorem}


\begin{proof}
Consider a quadratic storage function $\mathrm{V}_i(\tilde{x}_i) = \tilde{x}_i^\T P_i \tilde{x}_i$, where $P_i > 0$. The derivative of the storage function along the system trajectory can be evaluated as:
\begin{equation}\label{Eq:Th:Local_CPLStep1}
\begin{aligned}
\dot{\mathrm{V}}_i(\tilde{x}_i) &= 2\tilde{x}_i^\T P_i(\hat{A}_i\tilde{x}_i + g_i(\tilde{x}_i) + \tilde{u}_i)\\
&= 2\tilde{x}_i^\T P_i\hat{A}_i\tilde{x}_i + 2\tilde{x}_i^\T P_ig_i(\tilde{x}_i) + 2\tilde{x}_i^\T P_i\tilde{u}_i\\
&\leq \tilde{x}_i^\T \left(\mathcal{H}(P_iA_i + P_i B_i K_{i0}) + R_i\right)\tilde{x}_i + 2\tilde{x}_i^\T P_i\tilde{u}_i,
\end{aligned}
\end{equation}
where $\hat{A}_i \triangleq A_i + B_i K_{i0}$ and \eqref{Eq:Lm:S-Procedure2} have been used. Note that, to use \eqref{Eq:Lm:S-Procedure2}, based on Lm. \ref{Lm:S-ProcedureModified}, we require the inclusion of the constraint \eqref{Eq:Co:S-Procedure1}. This constraint translates to the last LMI constraint in \eqref{Eq:Th:Local_CPLStep1} under the change of variables $\tilde{R}_i \triangleq R_i^{-1}$, $\tilde{P}_i \triangleq P_i^{-1}$ and $\tilde{\lambda}_i \triangleq \frac{1}{\lambda_i}$. 

For the dissipativity property IF-OFP($\nu_i$, $\rho_i$) (from input $\tilde{u}_i$ to output $\tilde{x}_i$), we require:
\begin{equation}
\dot{\mathrm{V}}_i(\tilde{x}_i) \leq \begin{bmatrix} \tilde{u}_i \\ \tilde{x}_i \end{bmatrix}^\T \begin{bmatrix} -\nu_i\I & \frac{1}{2}\I \\ \frac{1}{2}\I & -\rho_i\I \end{bmatrix} \begin{bmatrix} \tilde{u}_i \\ \tilde{x}_i \end{bmatrix}.
\end{equation}
Applying \eqref{Eq:Th:Local_CPLStep1} and rearranging leads to the condition:
$$
\bm{\tilde{x}_i \\ \tilde{u}_i}^\T \Phi_i \bm{\tilde{x}_i \\  \tilde{u}_i} \geq 0 \iff  \Phi_ \geq 0,
$$
where $\Phi_i$ takes the form 
\begin{align*}
\Phi_i \triangleq \bm{
-\mathcal{H}(P_iA_i + P_i B_i K_{i0}) - (R_i + \rho_i\I)  & \frac{1}{2}\I - P_i \\ 
\frac{1}{2}\I - P_i & -\nu_i \I} \geq 0.
\end{align*}

Since $P_i$ and $K_{i0}$ appear in a bilinear manner in the above obtained $\Phi_i$ expression, we pre-and post multiply $\Phi_i$ with $\diag([P_i^{-1}, \I])$ and apply a change of variables $\tilde{P}_i \triangleq P_i^{-1}$ and $\tilde{K}_{i0} \triangleq K_{i0}\tilde{P}_i$, to obtain an equivalent condition for $\Phi_i \geq 0$ as 
\begin{equation*}
\bm{
-\mathcal{H}(A_i\tilde{P}_i + B_i \tilde{K}_{i0}) - \tilde{P}_i(R_i + \rho_i\I)\tilde{P}_i  & \frac{1}{2}\tilde{P}_i - \I \\ 
\frac{1}{2}\tilde{P}_i - \I & -\nu_i \I
} \geq 0.   
\end{equation*}

Next, assuming $R_i + \rho_i \I > 0$ (as justified in As. \ref{As:NegativeDissipativity} and Lm. \ref{Lm:S-ProcedureModified}) and applying Lm. \ref{Lm:Schur1}, we get an equivalent condition for $\Phi_i \geq 0$ as 
\begin{equation} \label{Eq:Th:Local_CPLStep2}
\bm{
(R_i + \rho_i\I)^{-1} & \tilde{P}_i & \0 \\
\tilde{P}_i & -\mathcal{H}(A_i\tilde{P}_i + B_i \tilde{K}_{i0}) & \frac{1}{2}\tilde{P}_i - \I \\ 
\0 & \frac{1}{2}\tilde{P}_i - \I & -\nu_i \I
} \geq 0. 
\end{equation}
Using the result established in Lm. \ref{Lm:Woodbury}, we get,  
$$
(R_i + \rho_i I)^{-1} = R_i^{-1} - R_i^{-1} \left( \rho_i^{-1} \I + R_i^{-1} \right)^{-1} R_i^{-1}. 
$$
This, when applied to \eqref{Eq:Th:Local_CPLStep2}, takes a form where  Lm. \ref{Lm:Schur1} can be re-applied. Consequently, we can obtain an equivalent condition for $\Phi_i \geq 0$ as 
\begin{equation}
\bm{
\rho_i^{-1} \I + R_i^{-1} &  R_i^{-1} & \0 & \0 \\
R_i^{-1} & R_i^{-1} & \tilde{P}_i & \0 \\
\0 & \tilde{P}_i & -\mathcal{H}(A_i\tilde{P}_i + B_i \tilde{K}_{i0}) & \frac{1}{2}\tilde{P}_i - \I \\ 
\0 & \0 & \frac{1}{2}\tilde{P}_i - \I & -\nu_i \I 
} \geq 0. 
\end{equation}
Finally, using the said change of variables, the above constraint leads to the main LMI constraint in \eqref{Eq:Th:Local_CPL}, which completes the proof.
\end{proof}

\subsection{Global Control and Topology Co-Design}

The local controllers \eqref{Controller} regulate the voltage at each DG while ensuring that closed-loop DG dynamics satisfy the required dissipativity properties established in Sec. \ref{Sec:Diss_Property}. Given these subsystem properties, we now synthesize the interconnection matrix $M$ \eqref{Eq:NetErrSysMMat} (see Fig. \ref{Fig.DissNetError}), particularly its block $K$, using Prop. \ref{synthesizeM}. Note that, by synthesizing $K=[K_{ij}]_{i,j\in\N_N}$, we can uniquely determine the consensus-based distributed global controller gains $\{k_{ij}^I:i,j\in\mathbb{N}_N\}$ \eqref{k_ij} (required in \eqref{ControllerG} to ensure the current sharing goal), along with the required communication topology $\mathcal{G}^c$. Note also that, when designing $K$ via Prop. \ref{synthesizeM}, we particularly enforce the closed-loop DC MG error dynamics to be $\textbf{Y}$-dissipative with $\textbf{Y} \triangleq \scriptsize \bm{\gamma^2\I & 0 \\ 0 & -\I}$ (see Rm. \ref{Rm:X-DissipativityVersions}) to prevent/bound the amplification of disturbances affecting the performance (voltage regulation and current sharing). The following theorem formulates this distributed global controller and communication topology co-design problem. 

\begin{theorem}\label{Th:CentralizedTopologyDesign}
The closed-loop networked error dynamics of the DC MG (see in Fig. \ref{Fig.DissNetError}) can be made finite-gain $L_2$-stable with an $L_2$-gain $\gamma$ (where $\Tilde{\gamma} \triangleq \gamma^2<\bar{\gamma}$ and $\bar{\gamma}$ is prespecified) from unknown disturbances $w_c(t)$ to performance output $z_c(t)$, by synthesizing the interconnection matrix block $M_{\tilde{u}x}=K$ (\ref{Eq:MMatrix}) via solving the LMI problem:
\begin{equation}
\label{Eq:Th:CentralizedTopologyDesign0}
\begin{aligned}
\min_{\substack{Q,\{p_i: i\in\N_N\},\\
\{\bar{p}_l: l\in\N_L\}, \tilde{\gamma}, S}} &\sum_{i,j\in\N_N} c_{ij} \Vert Q_{ij} \Vert_1 +c_1 \tilde{\gamma} + \alpha\text{tr}(S), \\
\mbox{Sub. to:}\ &p_i > 0,\ \forall i\in\N_N,\ 
\bar{p}_l > 0,\ \forall l\in\N_L,\\   
\ &\mbox{\eqref{globalcontrollertheorem}: } W + S > 0, \ S \geq 0,\ \text{tr}(S) \leq \eta,\\
\ &0 < \tilde{\gamma} < \bar{\gamma},\\
& Q_I P_n \textbf{1}_N = 0,
\end{aligned}
\end{equation}
as $K = (\textbf{X}_p^{11})^{-1} Q$ and $Q_I = \bm{Q_{ij}^{2,2}}_{i,j\in\N_N}$, where 
$\textbf{X}^{12} \triangleq 
\diag([-\frac{1}{2\nu_i}\I]_{i\in\N_N})$, 
$\textbf{X}^{21} \triangleq (\textbf{X}^{12})^\T$,
$\Bar{\textbf{X}}^{12} \triangleq 
\diag([-\frac{1}{2\Bar{\nu}_l}\I]_{l\in\N_L})$,
$\Bar{\textbf{X}}^{21} \triangleq (\Bar{\textbf{X}}^{12})^\T$, 
$\textbf{X}_p^{11} \triangleq 
\diag([-p_i\nu_i\I]_{i\in\N_N})$, 
$\textbf{X}_p^{22} \triangleq 
\diag([-p_i\rho_i\I]_{i\in\N_N})$, 
$\Bar{\textbf{X}}_{\bar{p}}^{11} 
\triangleq \diag([-\bar{p}_l\bar{\nu}_l\I]_{l\in\N_L})$, 
$\Bar{\textbf{X}}_{\bar{p}}^{22} 
\triangleq \diag([-\bar{p}_l\bar{\rho}_l\I]_{l\in\N_L})$, and $\tilde{\Gamma} \triangleq \tilde{\gamma}\I$. 
The structure of $Q\triangleq[Q_{ij}]_{i,j\in\N_N}$ mirrors that of $K\triangleq[K_{ij}]_{i,j\in\N_N}$ (i.e., only the middle element is non-zero in each block $Q_{ij}$, see \eqref{k_ij}). 
The coefficients $c_1>0$ and $c_{ij}>0,\forall i,j\in\N_N$ are predefined cost coefficients corresponding to the $L_2$-gain (control cost) and communication links (communication cost), respectively. The matrix $S$ is a slack matrix included for numerical stability of the used LMI solver, where the slack coefficients $\alpha \geq 0$ and $\eta \geq 0$ respectively impose soft and hard constraints on $S$.
\end{theorem}


\begin{proof}
The proof follows by considering the closed-loop DC MG (shown in Fig. \ref{Fig.DissNetError}) as a networked system and applying the subsystem dissipativity properties assumed in \eqref{Eq:XEID_DG} and \eqref{Eq:XEID_Line} to the interconnection topology synthesis result given in Prop. \ref{synthesizeM}. We model the DG error subsystems as IF-OFP($\nu_i,\rho_i$) and line error subsystems as  IF-OFP($\Bar{\nu}_l,\Bar{\rho}_l$), which are secured through local controller design and analysis in Th. \ref{Th:Local_CPL} and Lm. \ref{Lm:LineDissipativityStep}, respectively.
The LMI problem \eqref{Eq:Th:CentralizedTopologyDesign0} is formulated to ensure the networked error system is \textbf{Y}-dissipative, thereby ensuring finite-gain $L_2$-stability with gain $\gamma$ from disturbances $w_c$ to performance outputs $z_c$. The objective function in \eqref{Eq:Th:CentralizedTopologyDesign0} consists of three terms: communication cost ($\sum_{i,j\in\N_N} c_{ij} \Vert Q_{ij} \Vert_1$), control cost $c_1 \tilde{\gamma}$, and numerical stability term ($\alpha\text{tr}(S)$). Minimizing this function while satisfying LMI constraints simultaneously optimizes the communication topology (by synthesizing $K = (\textbf{X}_p^{11})^{-1} Q$) and robust stability (by minimizing $\tilde{\gamma}$) while ensuring the given specification $\gamma^2<\bar{\gamma}$. The resulting controller and topology achieve voltage regulation and current sharing in the presence of ZIP loads and disturbances.    
\end{proof}

\begin{figure*}[!hb]
\vspace{-5mm}
\centering
\hrulefill
\begin{equation}\label{globalcontrollertheorem}
\scriptsize
	W \triangleq \bm{
		\textbf{X}_p^{11} & \0 & \0 & Q & \textbf{X}_p^{11}\Bar{C} &  \textbf{X}_p^{11}E_c \\
		\0 & \bar{\textbf{X}}_{\bar{p}}^{11} & \0 & \Bar{\textbf{X}}_{\Bar{p}}^{11}C & \0 & \bar{\textbf{X}}_{\bar{p}}^{11}\bar{E}_c \\
		\0 & \0 & \I & H_c & \bar{H}_c & \0 \\
		Q^\T & C^\T\Bar{\textbf{X}}_{\Bar{p}}^{11} & H_c^\T & -Q^\T\textbf{X}^{12}-\textbf{X}^{21}Q-\textbf{X}_p^{22} & -\textbf{X}^{21}\textbf{X}_{p}^{11}\bar{C}-C^\T\bar{\textbf{X}}_{\bar{p}}^{11}\bar{\textbf{X}}^{12} & -\textbf{X}^{21}\textbf{X}_p^{11}E_c \\
		\Bar{C}^\T\textbf{X}_p^{11} & \0 & \bar{H}_c^\T & -\Bar{C}^\T\textbf{X}_p^{11}\textbf{X}^{12}-\bar{\textbf{X}}^{21}\Bar{\textbf{X}}_{\Bar{p}}^{11}C & -\bar{\textbf{X}}_{\bar{p}}^{22} & -\bar{\textbf{X}}^{21}\Bar{\textbf{X}}_{\Bar{p}}^{11}\bar{E}_c \\ 
		E_c^\T\textbf{X}_p^{11} & \bar{E}_c^\T\Bar{\textbf{X}}_{\Bar{p}}^{11} & \0 & -E_c^\T\textbf{X}_p^{11}\textbf{X}^{12} & -\bar{E}_c^\T\Bar{\textbf{X}}_{\Bar{p}}^{11}\bar{\textbf{X}}^{12} & \tilde{\Gamma} \\
	}\normalsize 
 >0 
\end{equation}
\end{figure*}

\begin{figure*}[!hb]
\vspace{-5mm}
\centering
\begin{equation}\label{Eq:Neccessary_condition}
\scriptsize
	\bm{
		-p_i\nu_i & 0 & 0 & 0 & -p_i\nu_i\bar{C}_{il} & -p_i\nu_i\\
		0 & -\bar{p}_l\bar{\nu}_l & 0 & -\bar{p}_l\bar{\nu}_lC_{il} & 0 & -\bar{p}_l\bar{\nu}_l \\
		0 & 0 & 1 & 1 & 1 & 0 \\
		0 & -C_{il}\bar{\nu}_l\bar{p}_l & 1 & p_i\rho_i & -\frac{1}{2}p_i\bar{C}_{il}-\frac{1}{2}C_{il}\bar{p}_l & -\frac{1}{2}p_i \\
		-\bar{C}_{il}\nu_ip_i & 0 & 1 & -\frac{1}{2}\bar{C}_{il}p_i-\frac{1}{2}\bar{p}_lC_{il} & \bar{p}_l\bar{\rho}_l & -\frac{1}{2}p_l \\ 
		-\nu_ip_i & -\bar{p}_l\bar{\nu}_l & 0 & -\frac{1}{2}p_i & -\frac{1}{2}\bar{p}_l & \tilde{\gamma}_i \\
	}\normalsize 
 >0,\ \forall l\in \mathcal{E}_i, \forall i\in\N_N
\end{equation}
\end{figure*}

\begin{remark}
In the proposed co-design approach \eqref{Eq:Th:CentralizedTopologyDesign0}: (i) communication costs are minimized through sparse topology optimization, (ii) control performance is improved by minimizing the $L_2$-gain from disturbance inputs to performance outputs, and (iii) computational efficiency is not compromised through LMI formulation.    
\end{remark}

\subsection{Necessary Conditions on Subsystem Passivity Indices}

Based on the terms $\textbf{X}_p^{11}$, $\textbf{X}_p^{22}$, $\bar{\textbf{X}}_{\bar{p}}^{11}$, $\bar{\textbf{X}}_{\bar{p}}^{22}$, $\textbf{X}^{12}$, $\textbf{X}^{21}$, $\bar{\textbf{X}}^{12}$, and $\bar{\textbf{X}}^{21}$ appearing in \eqref{globalcontrollertheorem} included in the global co-design problem \eqref{Eq:Th:CentralizedTopologyDesign0}, it is clear that the feasibility and the effectiveness of the proposed global co-design technique (i.e., Th. \ref{Th:CentralizedTopologyDesign}) depends on the enforced passivity indices $\{(\nu_i,\rho_i):i\in\mathbb{N}_N\}$ \eqref{Eq:XEID_DG} and  $\{(\bar{\nu}_l,\bar{\rho}_l):l\in\mathbb{N}_L\}$ \eqref{Eq:XEID_Line} assumed for the DG error dynamics \eqref{Eq:DG_error_dynamic} and line error dynamics \eqref{Eq:Line_error_dynamic}, respectively.

However, using Th. \ref{Th:Local_CPL} for designing dissipativating local controllers in $\{u_{iL}:i\in\mathbb{N}_N\}$ \eqref{Controller}, we can obtain a specialized set of passivity indices for the DG error dynamics \eqref{Eq:DG_error_dynamic}. Similarly, using Lm. \ref{Lm:LineDissipativityStep} for dissipativity analyses, we can obtain a specialized set of passivity indices for the line error dynamics \eqref{Eq:Line_error_dynamic}. Hence, these local controller design and analysis processes have a great potential to impact the feasibility and effectiveness of the global co-design solution.

Therefore, when designing such local controllers (via Th. \ref{Th:Local_CPL}) and conducting such dissipativity analysis (via Lm. \ref{Lm:LineDissipativityStep}), one must also consider the specific conditions necessary for the feasibility and and implications on the effectiveness of the eventual global co-design solution. The following lemma, inspired by \cite[Lm. 1]{WelikalaJ22022}, identifies local necessary conditions based on the global LMI conditions \eqref{Eq:Th:CentralizedTopologyDesign0} in the global co-design problem in Th. \ref{Th:CentralizedTopologyDesign}.

\begin{lemma}\label{Lm:CodesignConditions}
For the LMI conditions \eqref{Eq:Th:CentralizedTopologyDesign0} in Th. \ref{Th:CentralizedTopologyDesign} to hold, it is necessary that the passivity indices $\{\nu_i,\rho_i:i\in\mathbb{N}_N\}$ \eqref{Eq:XEID_DG} and  $\{\bar{\nu}_l,\bar{\rho}_l:l\in\mathbb{N}_L\}$ \eqref{Eq:XEID_Line} respectively enforced for the DG \eqref{Eq:DG_error_dynamic} and line \eqref{Eq:Line_error_dynamic} error dynamics \eqref{Eq:Line_error_dynamic} are such that the LMI problem: 
\begin{equation}\label{Eq:Lm:CodesignConditions}
\begin{aligned}
\mbox{Find: }&\ \{(\nu_i,\rho_i,\tilde{\gamma}_i):i\in\N_N\},\{(\Bar{\nu}_l,\bar{\rho}_l):l\in\N_L\}\\
\mbox{Sub. to: }&\ 0 \leq \tilde{\gamma}_i \leq \bar{\gamma},\  \forall i\in\N_N,\ \eqref{Eq:Neccessary_condition},  
\end{aligned}
\end{equation} 
is feasible, where $p_i>0, \forall i\in\N_N$ and $\bar{p}_l>0, \forall l\in\N_L$ are some prespecified parameters. 
\end{lemma}

\begin{proof}
For the feasibility of the global co-design problem  \eqref{Eq:Th:CentralizedTopologyDesign0}, $W$ given in \eqref{globalcontrollertheorem} must satisfy $W > 0$. Let $W = [W_{rs}]_{r,s\in\N_6}$ where each block $W_{rs}$ can be a block matrix of block dimensions $(N\times N)$, $(N\times L)$ or $(L \times L)$ depending on its location in $W$ (e.g., see blocks $W_{11}, W_{15}$ and $W_{22}$, respectively). Without loss of generality, let us denote $W_{rs} \triangleq [W_{rs}^{jm}]_{j\in\bar{J}(r),m\in\bar{M}(s)}$ where $\bar{J}(r),\bar{M}(s) \in \{N,L\}$. Inspired by \cite[Lm. 1]{WelikalaJ22022}, we can obtain an equivalent condition for $W>0$ as $\bar{W} \triangleq \text{BEW}(W) > 0$ where $\text{BEW}(W)$ is the ``block-elementwise'' form of $W$, created by combining appropriate inner-block elements of each of the blocks $W_{rs}$ to create a $6\times 6$ block-block matrix. 
Simply, $\bar{W} = [[W_{rs}^{j,m}]_{r,s \in \N_6}]_{j \in \N_{\bar{J}},m\in\N_{\bar{M}}}$. Considering only the diagonal blocks in  $\bar{W}$ and the implication $\bar{W} > 0 \implies [[W_{rs}^{j,m}]_{r,s \in \N_6}]_{j \in \N_{\bar{J}(r)},m\in\N_{\bar{M}(r)}} > 0 \iff$\eqref{Eq:Lm:CodesignConditions} (also recall the notations $C_{il} \triangleq -C_{ti}\bar{C}_{il}^\T$, $\Bar{C}_{il} \triangleq -C_{ti}^{-1}$). Therefore,  \eqref{globalcontrollertheorem} $\implies$ \eqref{Eq:Lm:CodesignConditions}, in other words, \eqref{Eq:Lm:CodesignConditions} is a set of necessary conditions for the feasibility of the global co-design constraint \eqref{globalcontrollertheorem}. 

Besides merely supporting the feasibility of the global co-design \eqref{globalcontrollertheorem}, the LMI problem \eqref{Eq:Lm:CodesignConditions}, through its inclusion of the constraint $0 \leq \tilde{\gamma}_i \leq \bar{\gamma}$ (which can also be embedded in the objective function), inspires to improve the effectiveness (performance) of the global co-design \eqref{globalcontrollertheorem}.  
\end{proof}

In conclusion, here we used the LMI problem \eqref{Eq:Th:CentralizedTopologyDesign0} to derive a set of necessary LMI conditions consolidated as a single LMI problem \eqref{Eq:Lm:CodesignConditions}. Ensuring the feasibility of this consolidated LMI problem \eqref{Eq:Lm:CodesignConditions} increases the feasibility and effectiveness of the LMI problem \eqref{Eq:Th:CentralizedTopologyDesign0} solution, i.e., of the global co-design. Finally, we also point out that the necessary conditions given in the LMI problem \eqref{Eq:Lm:CodesignConditions} are much stronger and complete than those given in our prior work \cite{ACCNajafi}.

\subsection{Local Controller Synthesis}\label{Sec:Local_Synth}

We conclude our proposed solution by providing the following theorem that integrates all the necessary LMI conditions for the global co-design of the DC MG (i.e., Th. \ref{Th:CentralizedTopologyDesign}), identified in Lm. \ref{Lm:CodesignConditions}, and use them simultaneously to design the local controllers for DG error dynamics and analyze local line error dynamics. In all, the following result removes the necessity of implementing/evaluating the LMI problems in Th. \ref{Th:Local_CPL}, Lm. \ref{Lm:LineDissipativityStep} and Lm. \ref{Lm:CodesignConditions} separately, and instead provides a unified LMI problem to lay the foundation required to execute the global control and topology co-design of the DC MG using the established Th. \ref{Th:CentralizedTopologyDesign}.

\begin{theorem}\label{Th:LocalControllerDesign}
Under the predefined DG parameters \eqref{Eq:DGCompact}, line parameters \eqref{Eq:LineCompact} and design parameters $\{p_i: i\in\N_N\}$, $\{\bar{p}_l:l\in\N_L\}$, the necessary conditions in \eqref{Eq:Th:CentralizedTopologyDesign0} hold if the local controller gains $\{K_{i0}, i\in\N_N\}$ (\ref{Controller}) and 
DG and line passivity indices 
$\{\nu_i,\tilde{\rho}_i:i\in\mathbb{N}_N\}$ \eqref{Eq:XEID_DG} and $\{\bar{\nu}_l,\bar{\rho}_l:l\in\mathbb{N}_L\}$ \eqref{Eq:XEID_Line} are determined by solving the LMI problem:
\begin{equation}
\begin{aligned}\nonumber
&\min_{\tilde{\lambda}_i} \sum_{i=1}^{N} \alpha_{\tilde{\lambda}}\tilde{\lambda}_i, \\
&\mbox{Find: }\ 
\{(\tilde{K}_{i0}, \tilde{P}_i, \tilde{R}_i, \tilde{\lambda}_i, \nu_i, \tilde{\rho}_i, \tilde{\gamma}_i): i\in\N_N\},\\ 
&\quad \quad \quad \{(\bar{P}_l, \bar{\nu}_l,\bar{\rho}_l): l\in\mathbb{N}_L\},\ \{(\xi_{il},s_1,s_2): l \in \mathcal{E}_i, i\in\N_N\}\\
&\mbox{Sub. to: }\ 
\tilde{P}_i > 0,\ \tilde{R}_i > 0,\ \tilde{\lambda}_i > 1,\ \bar{P}_l > 0,\ \eqref{Eq:Transformed_Necessary_condition},\\   
&\begin{bmatrix}
\tilde{\rho}_i\I + \tilde{R}_i & \tilde{R}_i & \0 & \0 \\
\tilde{R}_i & \tilde{R}_i & \tilde{P}_i & \0 \\
\0 & \tilde{P}_i & -\mathcal{H}(A_i\tilde{P}_i + B_i\tilde{K}_{i0}) & -\I + \frac{1}{2}\tilde{P}_i \\ 
\0 & \0 & -\I + \frac{1}{2}\tilde{P}_i & -\nu_i\I    
\end{bmatrix} > 0,\\
& \bm{\I & \0 & \0 \\ \0 & \tilde{P}_i T_i & \0} 
+ \bm{\I &\ 0\\ \0 & T_i^\T \tilde{P}_i \\ \0 & \0} 
-\bm{\tilde{R}_i & \0 \\ \0 & \I } \\
&- \bm{\Theta_{11} & \Theta_{12}\tilde{P}_i + \tilde{\lambda}_i\I \\
\tilde{P}_i\Theta_{21} + \tilde{\lambda}_i\I & \0} 
\geq 0, \ \forall i\in\N_N,\\
&\ \bm{
\frac{2\bar{P}_lR_l}{L_l}-\bar{\rho}_l & -\frac{\bar{P}_l}{L_l}+\frac{1}{2}\\
    -\frac{\bar{P}_l}{L_l}+\frac{1}{2} & -\bar{\nu}_l
} 
\geq0,\ \forall l \in \N_L,\   \\
&\begin{bmatrix} 
1 & \bar{\nu}_l & \tilde{\rho}_i \\
\bar{\nu}_l & s_1 & \xi_{il} \\
\tilde{\rho}_i & \xi_{il} & s_2
\end{bmatrix} \geq 0,\ \forall l \in \mathcal{E}_i, \forall i \in \mathbb{N}_N,
\end{aligned}
\end{equation}
where $K_{i0} \triangleq \tilde{K}_{i0}\tilde{P}_i^{-1}$, $\rho_i = 1/\tilde{\rho}_i$, $\xi_{il}$ is an auxiliary variable, $\alpha_\lambda$ is a positive weighting coefficient and $s_1$ and $s_2$ are semidefinite optimization variables that ensure the positive semidefiniteness of the Schur complement matrix while enforcing the relationship between $\xi_{il}$ and bilinear terms $\bar{\nu}_l\tilde{\rho}_i$.
\end{theorem} 

\begin{figure*}[!hb]
\vspace{-5mm}
\centering
\hrulefill
\begin{equation}\label{Eq:Transformed_Necessary_condition}
\scriptsize
	\bm{
		-p_i\nu_i & 0 & 0 & 0 & -p_i\nu_i\bar{C}_{il} & -p_i\nu_i\\
		0 & -\bar{p}_l\bar{\nu}_l & 0 & -\bar{p}_l\xi_{il}C_{il} & 0 & -\bar{p}_l\bar{\nu}_l \\
		0 & 0 & 1 & \tilde{\rho}_i & 1 & 0 \\
		0 & -C_{il}\xi_{il}\bar{p}_l & \tilde{\rho}_i & p_i\tilde{\rho}_i & -\frac{1}{2}p_i\bar{C}_{il}\tilde{\rho}_i-\frac{1}{2}C_{il}\bar{p}_l\tilde{\rho}_i & -\frac{1}{2}p_i\tilde{\rho}_i \\
		-\bar{C}_{il}\nu_ip_i & 0 & 1 & -\frac{1}{2}\bar{C}_{il}p_i\tilde{\rho}_i-\frac{1}{2}\bar{p}_lC_{il}\tilde{\rho}_i & \bar{p}_l\bar{\rho}_l & -\frac{1}{2}p_l \\ 
		-\nu_ip_i & -\bar{p}_l\bar{\nu}_l & 0 & -\frac{1}{2}p_i\tilde{\rho}_i & -\frac{1}{2}\bar{p}_l & \tilde{\gamma}_i \\
	}\normalsize 
 >0,\ \forall l\in \mathcal{E}_i, \forall i\in\N_N
\end{equation}
\end{figure*}

\begin{proof}
The proof proceeds as follows: (i) We start by considering the dynamic models of $\tilde{\Sigma}_i^{DG}$ and $\tilde{\Sigma}_i^{Line}$ as described in \eqref{Eq:DG_error_dynamic} and \eqref{Eq:Line_error_dynamic}, respectively. We then apply the LMI-based controller synthesis and analysis techniques from Th. \ref{Th:Local_CPL} and Lm. \ref{Lm:LineDissipativityStep} to enforce and identify the subsystem passivity indices assumed in \eqref{Eq:XEID_DG} and \eqref{Eq:XEID_Line}, respectively. (ii) Next, we apply the LMI formulations from Th. \ref{Th:Local_CPL} to obtain the local controller gains $K_{i0}$ and passivity indices $(\nu_i,\rho_i)$ for $\tilde{\Sigma}_i^{DG}$, which leads to the constraints in the first part of the LMI problem in Th. \ref{Th:LocalControllerDesign}. Similarly, we apply Lm. \ref{Lm:LineDissipativityStep} to identify the passivity indices $(\bar{\nu}_l,\bar{\rho}_l)$ for $\tilde{\Sigma}_i^{Line}$, which leads to the constraints in the second part of the LMI problem. (iii) To handle the transformation from $\rho_i$ to $\tilde{\rho}_i$ where $\tilde{\rho}_i = \rho_i^{-1}$, we apply a congruence transformation to equation \eqref{Eq:Neccessary_condition}. Using the transformation matrix $T = \diag(1, 1, 1, 1/\rho_i, 1, 1)$, the   term $p_i\rho_i$ at position (4,4) of \eqref{Eq:Neccessary_condition} becomes $p_i\tilde{\rho}_i$. This transformation also affects other elements in row 4 and column 4, resulting in the transformed equation \eqref{Eq:Transformed_Necessary_condition}. The bilinear terms involving $\bar{\nu}_l\tilde{\rho}_i$ in positions (2,4), (4,2), and other locations in row 4 and column 4 require special handling. For each bilinear term $\bar{\nu}_l\tilde{\rho}_i$, we introduce an auxiliary variable $\xi_{il}$ and add the Schur complement constraint:
\begin{equation*}
\begin{bmatrix} 
1 & \bar{\nu}_l & \tilde{\rho}_i \\
\bar{\nu}_l & s_1 & \xi_{il} \\
\tilde{\rho}_i & \xi_{il} & s_2
\end{bmatrix} \geq 0,
\end{equation*}

where $s_1$ and $s_2$ are semidefinite variables determined during optimization. This constraint enforces this $\xi_{il} \geq \bar{\nu}_l\tilde{\rho}_i$, allowing us to replace each bilinear term with $\xi_{il}$ in the transformed equation. (iv) Finally, we impose the necessary conditions on subsystem passivity indices identified in Lm. \ref{Lm:CodesignConditions} to support the feasibility and effectiveness of the global control and communication topology co-design approach presented in Th. \ref{Th:CentralizedTopologyDesign}. The constraint \eqref{Eq:Transformed_Necessary_condition} is derived from the necessary condition established in Lm. \ref{Lm:CodesignConditions}, ensuring that the local controllers and their passivity indices are compatible with the global co-design problem. By formulating this unified LMI problem, we obtain a one-shot approach to simultaneously design local controllers and determine passivity indices that guarantee the feasibility of the global co-design problem.
\end{proof}

\section{Simulation Results}\label{Simulation}

We conducted simulations using an islanded DC MG configuration to evaluate the effectiveness of the proposed dissipativity-based control framework for DC MGs with ZIP loads. The test system comprises 4 DGs with ZIP loads interconnected through 4 transmission lines, implemented in a MATLAB/Simulink environment. The voltage source converters have a nominal voltage of $120$ V, while the desired reference voltage is set at $V_r = 48$ V. Detailed parameters of the simulated DC MG are available in \cite[Tab. 1]{Najafirad2024Ax1}. Figure \ref{fig.physicalcommunicationtopology} illustrates the physical and communication topology of the DC MG. For the controller and topology co-design procedure, we selected parameters $p_i = 0.1, \forall i \in \N_N$ and $p_l = 0.01, \forall l \in \N_L$.

\begin{figure}
    \centering
\includegraphics[width=0.99\columnwidth]{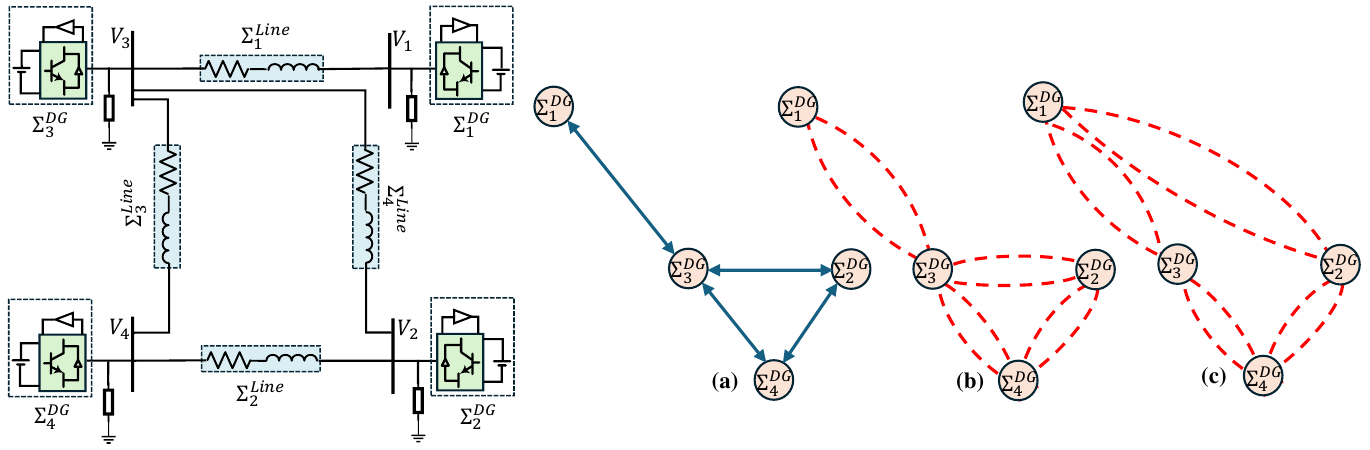}
    \caption{The topology of the islanded DC MG with 4 DGs and 4 lines: (a) Physical topology, (b) Co-designed communication topology under hard graph constraints, and (c) Co-designed communication topology under soft graph constraints.}
    \label{fig.physicalcommunicationtopology}
\end{figure}

To evaluate the effectiveness of our proposed control methodology, we examined sequential load variations at each DG, where introduced constant current loads $\bar{I}_L$ at $t=1$ s, disconnected constant impedance loads $Y_L$ at $t=3$ s with subsequent reconnection at $t=5$ s, and added constant power loads $P_L$ at $t=8$ s. Fig. \ref{fig.voltagecurrent}(a) demonstrates that the output voltages successfully track the reference voltage $V_r$ despite these sequential load disturbances. The system's response to the CPL becomes evident at $t=8$ s, wherein the proposed local controller effectively manages the transient oscillations and restores voltage to the reference value. Furthermore, as illustrated in Fig. \ref{fig.voltagecurrent}(b), the implemented distributed global control architecture maintains proper current sharing among DGs even under challenging ZIP load conditions.

\begin{figure}
    \centering
    \includegraphics[width=0.99\columnwidth]{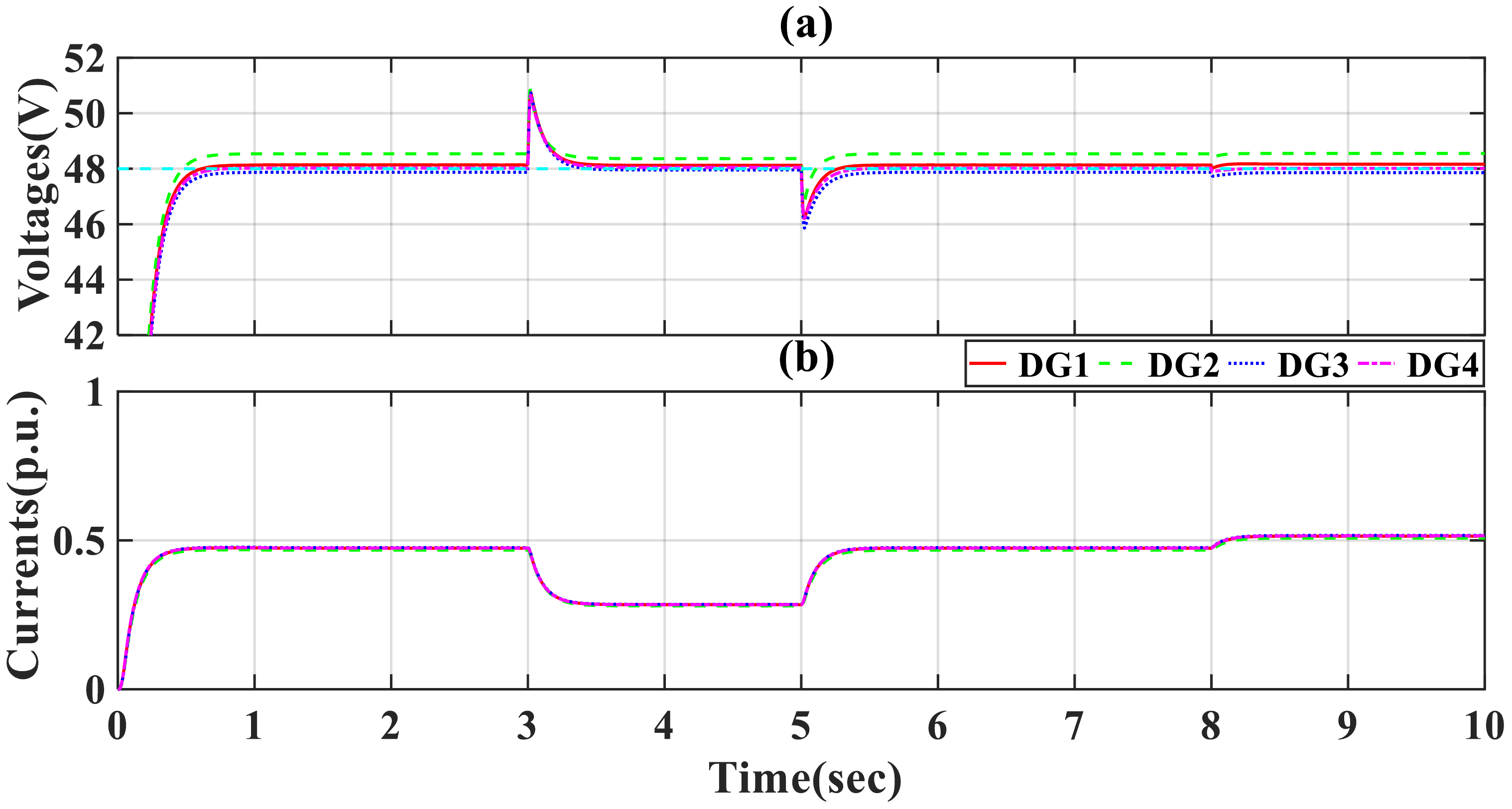}
    \caption{The DG outputs: (a) voltages and (b) per-unit currents, observed under dissipativity-based controller for the DC MG with ZIP loads shown in Fig. \ref{fig.physicalcommunicationtopology}(a).}
    \label{fig.voltagecurrent}
\end{figure}

This simulation study examines two distinct co-design methodologies, based on the usage of: (i) a hard graph constraint that requires strict alignment between the communication topology $\mathcal{G}^c$ and the physical topology $\mathcal{G}^p$, and (ii) a soft graph constraint that only introduces penalties for communication links that deviate from the physical network architecture $\mathcal{G}^p$. Figure \ref{fig.physicalcommunicationtopology}(a) presents the physical topology of the DC MG test system. Fig. \ref{fig.physicalcommunicationtopology}(b) illustrates how the hard graph constraint adheres precisely to the physical topology configuration. However, the proposed dissipativity-based controller, under the soft graph constraint, allows the resulting communication topology to deviate from the physical topology. As depicted in Fig. \ref{fig.physicalcommunicationtopology}(c), the soft graph constraint enables optimization of the communication topology to enhance the closed-loop robustness via improving information sharing/distribution among the DGs in the DC MG. This approach proves particularly beneficial for $\Sigma_1^{DG}$, which exhibits significant unique separation and thus requires additional information exchange with other DGs to facilitate proper current sharing coordination.

Finally, we compare our proposed dissipativity-based control approach with the conventional droop-based control method described in \cite{guo2018distributed}, focusing on voltage regulation performance. For fair comparison, both simulations used identical DG and line parameters, with optimally tuned droop control gains. As shown in Fig. \ref{fig.droop_control}, the dissipativity-based controller maintains accurate voltage regulation throughout the simulation. The droop controller exhibits larger overshoots during load changes, demonstrating a less effective response to load variations. The droop control also shows an initial voltage drop due to its inherent characteristics. This requires an additional control layer for compensation. Furthermore, when CPL introduces to the DC MG at $t=8$ s, the droop-based controller produces significant oscillations that compromise system performance, while the dissipativity-based controller handles the CPL effectively.

\begin{figure}
    \centering
\includegraphics[width=0.9\columnwidth]{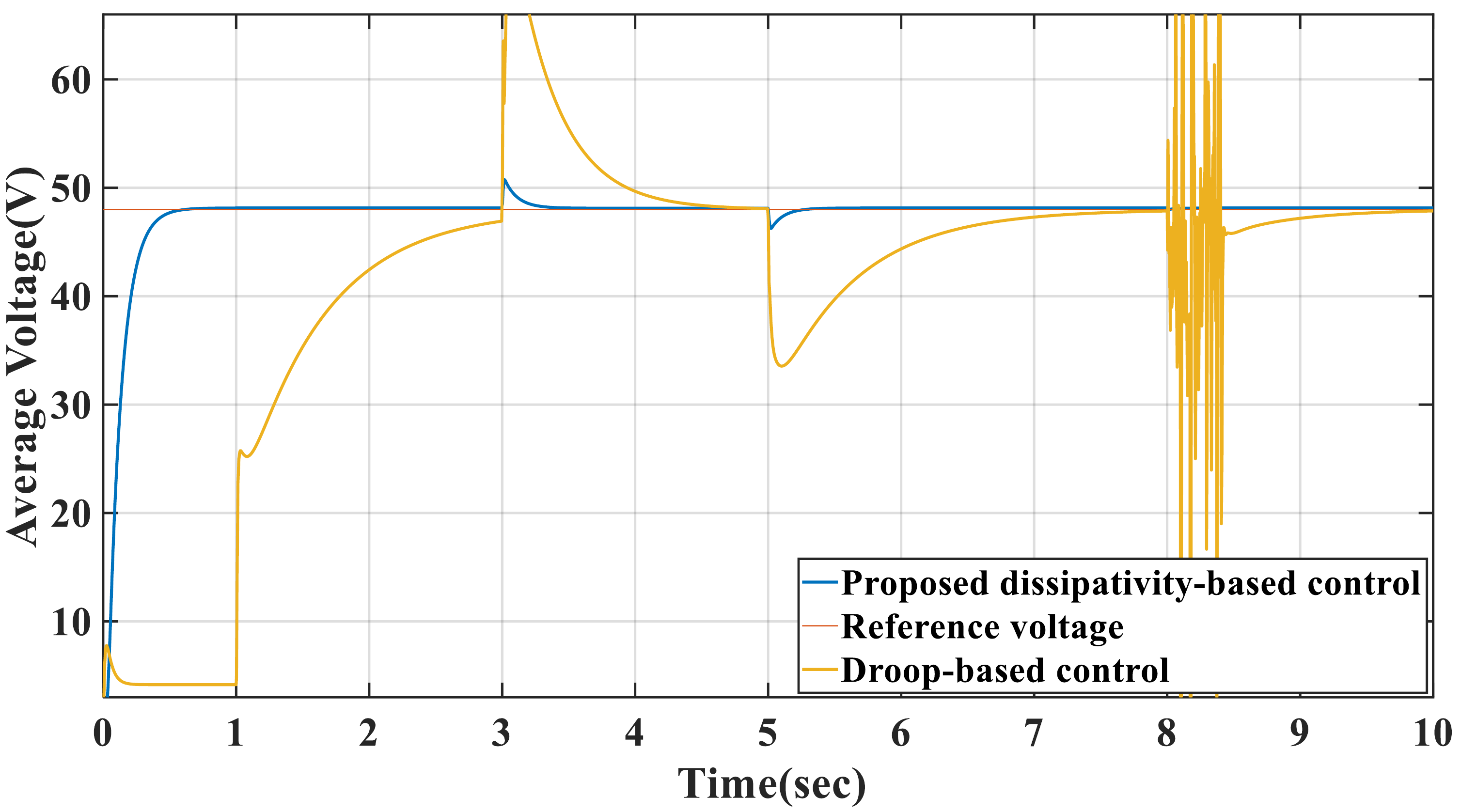}
    \caption{Comparison of average voltage regulation between the proposed dissipativity-based controller and droop controllers in the presence of ZIP laod for DC MG.}
    \label{fig.droop_control}
\end{figure}

\section{Conclusion}\label{Conclusion}
This paper presents a dissipativity-based distributed control and topology co-design approach for DC microgrids that addresses voltage regulation, current sharing, and generic ZIP loads. By leveraging dissipativity and sector-boundedness concepts, we develop a unified framework to co-design voltage reference levels, local steady state controllers, local feedback controllers, global distributed controllers, along with a communication topology so as to ensure dissipativity of the closed-loop DC MG, from generic disturbance inputs to voltage regulation and current sharing performance outputs. Unlike conventional droop-based methods, our approach eliminates the need for precise droop coefficient tuning, enhancing voltage regulation accuracy while maintaining proportional current sharing. Moreover, the proposed approach in this paper is LMI based, and thus can be conveniently implemented and efficiently and accurately evaluated using existing standard convex optimization tools. Simulation results have shown the effectiveness and the superior performance compared to conventional droop control methods, particularly when handling CPLs. Future work will focus on developing plug-and-play capabilities and extending the approach to more complex microgrid dynamics.

\bibliographystyle{IEEEtran}
\bibliography{References}

\end{document}